\documentclass[a4paper,11pt]{amsart}

\usepackage{latexsym}
\usepackage{amssymb,amsmath}
\usepackage{graphicx, caption, chngpage}
\usepackage{url}
\usepackage[vcentermath]{youngtab}
  
\usepackage{booktabs,footmisc} 
\usepackage{tikz}
\usepackage{qcircuit}

\usepackage{zx-calculus}

\allowdisplaybreaks
\newtheorem{theorem}{Theorem}[section]
\newtheorem*{theorem*}{Theorem}

\newtheorem{lemma}[theorem]{Lemma}
\theoremstyle{definition}
\newtheorem{definition}[theorem]{Definition}
\newtheorem{remark}[theorem]{Remark}

\newtheorem{example}[theorem]{Example}
\newtheorem{exercise}[theorem]{Exercise}

\newcommand{\transpose}{\mathrm{t}}
\DeclareMathOperator{\tr}{tr}

\DeclareMathOperator{\CNOT}{CNOT}
\DeclareMathOperator{\SWAP}{SWAP}

\newcommand{\pl}{{}\!\hskip-0.5pt+\hskip-0.5pt\!{}}
\newcommand{\mi}{{}\!\hskip-0.5pt-\hskip-0.5pt\!{}}

\newcommand{\Z}{\mathbb{Z}}

\newcommand{\R}{\mathbb{R}}

\newcommand{\PP}{\mathbb{P}}

\newcommand{\CC}{\mathcal{C}}

\newcommand{\C}{\mathbb{C}}

\newcommand{\F}{\mathbb{F}}

\renewcommand{\epsilon}{\varepsilon}

\DeclareMathOperator{\Sym}{Sym}

\renewcommand{\theta}{\vartheta}

\newcommand{\SU}{\mathrm{SU}}
\newcommand{\su}{\mathsf{su}}
\newcommand{\SO}{\mathrm{SO}}
\newcommand{\so}{\mathsf{so}}

\newcommand{\mfrac}[2]{{\textstyle\frac{#1}{#2}}}

\newcounter{thmlistcnt}
	{\setcounter{thmlistcnt}{0}%
	\begin{list}{\emph{(\roman{thmlistcnt})}}{%
		\usecounter{thmlistcnt}%
		\setlength{\topsep}{0pt}%
		\setlength{\leftmargin}{0pt}%
		\setlength{\itemsep}{0pt}%
		\setlength{\labelwidth}{17pt}
		\setlength{\itemindent}{30pt}}%
	}%
	{\end{list}}%
	
\newcounter{exlistcnt}
\newenvironment{exlist}%
	{\setcounter{exlistcnt}{0}%
	\begin{list}{(\alph{exlistcnt})}{%
		\usecounter{exlistcnt}%
		\setlength{\topsep}{0pt}%
		\setlength{\leftmargin}{0pt}%
		\setlength{\itemsep}{0pt}%
		\setlength{\labelwidth}{12pt}
		\setlength{\itemindent}{18pt}}%
	}%
	{\end{list}}%	

\newcounter{exlistBcnt}
\newenvironment{exlistB}%
	{\setcounter{exlistBcnt}{0}%
	\begin{list}{(\alph{exlistBcnt})}{%
		\usecounter{exlistBcnt}%
		\setlength{\topsep}{0pt}%
		\setlength{\leftmargin}{0pt}%
		\setlength{\itemsep}{0pt}%
		\setlength{\labelwidth}{17pt}
		\setlength{\itemindent}{30pt}}%
	}%
	{\end{list}}%
	
\newcommand{\ket}[1]{\left|#1\right>}
\newcommand{\bra}[1]{\left<#1\right|}
\newcommand{\braket}[2]{\left< #1|#2\right>}

\newcommand{\HH}{\mathcal{H}}

\numberwithin{equation}{section}
\usepackage{hyperref}

\newcommand{\vecn}[1]{{\hskip1pt\widehat{\mathbf{#1}}}}

\makeatletter
\let\old@footnotetext\@footnotetext
\renewcommand{\@footnotetext}[1]{%
    \old@footnotetext{%
    \linespread{0.98}\selectfont % <--- Adjust "0.8" to change spacing
    #1%
    }%
}
\makeatother

\newcommand{\rX}{{\red X}}
\newcommand{\rZ}{{\red Z}}

\renewcommand{\Re}{\hskip2pt\mathrm{Re}\,}
\renewcommand{\Im}{\hskip2pt\mathrm{Im}\,}

\newcommand{\red}[1]{\color[rgb]{1,0.0,0.0}#1}

\newcommand{\caret}{\texttt{\char`\^}}

\begin{document}

\begin{abstract}
These notes introduce quantum computation and quantum error correction, emphasising
the importance of stabilisers and the mathematical foundations in basic Lie theory.
We begin by using the double cover map $\mathrm{SU}_2 \rightarrow \mathrm{SO}_3(\mathbb{R})$
to illustrate the distinction between states and measurements for a single qubit.
We then discuss entanglement and CNOT gates, the Deutsch--Jozsa Problem, and finally 
quantum error correction, using the Steane $[[7,1,3]]$-code as the main example.
The necessary background physics of unitary evolution and Born rule measurements is
developed as needed. The circuit model is used throughout.
\end{abstract}

\title[Quantum computation and quantum error correction]{Quantum computation and quantum error correction: the theoretical minimum${}^\star$ \\[-2pt] \scalebox{0.8}{\textbf{${}^\star$with answers and other extras in the footnotes}}}
\author{Mark Wildon}
\date{February 2026, \emph{Email for comments:} \texttt{mark.wildon@bristol.ac.uk}}

\maketitle
\thispagestyle{empty}

% Add why unitary: probability conserved but really its Schrödinger, d|\psi>/dt = i/hbar Ht, one parameter subgroup, exp of anti-Hermitian
% is unitary

% No cloning: and apparent cloning of Z-basis states, e.g. \alpha|0> + |\beta|1>|0> -> \alpha |00> + \beta|11>.
% Maybe reference Quantum in Pictures

\vspace*{-12pt}

\noindent\textbf{Introduction.} Maybe you know that atoms bond together into molecules by sharing pairs of electrons.
In quantum language, electrons are qubits, or elements of $2$-dimensional Hilbert space,
and one type of bond is the Bell pair made by
applying a CNOT gate with control qubit $\ket{+}$ 
to a target qubit $\ket{0}$ (see \S\ref{subsec:twoQubits}).
In this article we use 
controlled gates, most importantly the CNOT gate, to prove  that quantum
computers are strictly more powerful than classical computers. 
We introduce the key idea
of `measuring a stabiliser' 
and show how it 
can be used to implement quantum error correction, making quantum computation a practical 
possibility.\footnote{\textbf{Lie:} the early parts of these notes are
crawling with \raisebox{-3pt}{\begin{tikzpicture} \draw[thick] (-0.3,-0.05)--(0.3,-0.05); \node at (0,0) {lies}; \end{tikzpicture}} 
\!\! deliberate oversimplifications, the worst of which are mentioned in numbered 
footnotes like this. I hope you will be impressed by the number so far:
elements of Hilbert space proportional by a non-zero scalar define
the same physical state;
qubits are not necessarily electrons; a more typical bond is the spin singlet
state $\mfrac{1}{\sqrt{2}}\bigl(\ket{01} - \ket{10})$, which is the Bell pair 
\smash{$\mfrac{1}{\sqrt{2}}\bigl( \ket{00} +\ket{11}\bigr)$} with the $XZ$-operator applied (see 
Exercise~\ref{ex:BellObservables});
quantum computers are only proved to be more powerful
than classical computers \emph{relative to an oracle}; error correction only
works if the  raw qubits are not \emph{too} error prone. }%Other models of EC do not need ancilla states 
% but we didn't actually say ancillae were the only way

\subsubsection*{Outline:} \S\ref{sec:oneQubit}, \S\ref{sec:twoQubits}, \S\ref{sec:manyQubits},
\S\ref{sec:tooManyQubits} deal respectively with one qubit (basics), two qubits (CNOT gates, entanglement),
many qubits (quantum computing, including the Quantum Discrete Fourier Transform 
for Boolean functions) and `too many' qubits (quantum error correction).

\section{One qubit}\label{sec:oneQubit}

\subsection{The absolute minimum}\label{subsec:oneQubitMinimum}
A qubit is an element of $2$-dimensional Hilbert space $\HH = \C^2$. We fix the
orthonormal basis $\ket{0}$, $\ket{1}$ of $\HH$. Thus
\[ \HH = \bigl\{ \alpha \!\ket{0} + \beta \!\ket{1} : \alpha, \beta \in \C \bigr\} \]
and $\braket{0}{0} = \braket{1}{1} = 1$, $\braket{0}{1} = \braket{1}{0} = 0$.

\begin{remark}
Do not be scared by the ket notation. As a rough guide, symbols of the same type, such as `$00$' and `$01$'
define orthonormal kets in the same basis, while
symbols of different types, such as `$0$' and `$+$', define kets in different bases and generally
are not orthogonal. For instance
the inner product of $\ket{0}$ and $\ket{+}$ is given by flipping $\ket{0}$ so that it becomes $\bra{0}$
and then putting it next to $\ket{+}$, getting $\braket{0}{+}$, which is of course \smash{$\mfrac{1}{\sqrt{2}}$}.
Since bras such as \smash{$\bra{0}$} are elements of the dual space \smash{$\HH^\star$},
our inner product must be defined so that
the conjugation is on the left-hand side:
$\braket{\alpha v}{\beta w} = \overline{\alpha}\beta \braket{v}{w}$.\footnote{\textbf{Aside:}
physicists will put anything and everything into a ket. For instance they like to write $\ket{x}$ for the state `the particle is certainly at position $x \in \R$',
and will then happily apply an operator to this eigenstate
to extract the eigenvalue $x$. Sometimes
this operator is also denoted $x$, although 
it should be admitted that $\widehat{x}$ is
more common. Thus
$\widehat{x} \ket{x} = x \ket{x}$ is not a tautological string of barely
meaningful symbols, but instead a deep physical truth about measurement of quantum states.
Worse still, in the continuous setting, $\ket{x}$ is not really defined at all,
since it would have to be some kind of Dirac delta `function', and so not a member
of any reasonable Hilbert space. None of this matters.}
\end{remark}

Quantum computers work by applying unitary linear maps to $\HH$ and its tensor powers. These unitary linear maps are
called \emph{gates}.\footnote{\textbf{Lie:} we will get to measurement shortly. This \emph{circuit model}
is not the only model for quantum computation, but all reasonable models define
the same computational class~$\mathsf{BQP}$. See \S\ref{subsec:BQP} for more on $\mathsf{BQP}$, footnote~\ref{footnote:SolovayKitaev} for the
Solovay--Kitaev Theorem and \S\ref{subsec:whyUnitary}
for why gates are unitary.} For instance, the $Z$-\emph{gate} and $X$-\emph{gate} act on the single qubit space $\mathcal{H}$
and have matrices 
\[ Z = \left( \begin{matrix} 1 & 0 \\ 0 & -1 \end{matrix} \right), \quad 
   X = \left( \begin{matrix} 0 & 1 \\ 1 & 0 \end{matrix}
\right) \]
in the basis $\ket{0}$, $\ket{1}$. Observe that $\ket{0}$ is a $+1$-eigenvector of $Z$ and $\ket{1}$ is a $-1$-eigenvector
of $Z$. For this reason, we call $\ket{0}, \ket{1}$ the \emph{$Z$-basis} of $\HH$.
Similarly 
\[ \ket{+} = \frac{1}{\sqrt{2}} \bigl( \ket{0} + \ket{1} \bigr),
\quad \ket{-} = \frac{1}{\sqrt{2}} \bigl( \ket{0} - \ket{1} \bigr) \]
form an orthonormal basis of $\HH$ of $X$-eigenvectors, with eigenvalues $+1$ and $-1$;
this is the \emph{$X$-basis}. The $X$-gate is the analogue of the classical NOT gate, flipping between
$\ket{0}$ and $\ket{1}$,
while the $Z$-gate, which introduces a phase from the minus sign in $Z\ket{1} = -\ket{1}$, 
has no analogue in classical computing.

The $Z$- and $X$-bases are switched by the unitary %but not special unitary 
\emph{Hadamard} gate
\[ H = \frac{1}{\sqrt{2}} \left( \begin{matrix}  1 & 1 \\ 1 & -1\end{matrix} \right). \]
Note that $H\ket{0} = \ket{+}$, $H\ket{1} = \ket{-}$, and 
$H\ket{+} = \ket{0}$ and $H\ket{-} = \ket{1}$. Thus~$H$
has order $2$.

\begin{definition}[$Z$-basis measurement]\label{defn:measureZ}
Let $\alpha\!\ket{0} + \beta\! \ket{1}$ be a qubit normalized so that
$|\alpha|^2 + |\beta|^2 = 1$.
\emph{Measuring}  $\alpha\!\ket{0} + \beta\! \ket{1}$ in the $Z$-basis projects
it to $\ket{0}$
with probability $|\alpha|^2$ and to $\ket{1}$ with probability $|\beta|^2$.
The measurement result is $0$ or $1$, respectively.
\end{definition}

This is a special case of the \emph{Born rule}. Note that we need normalized
states in order for the probabilities as defined to sum to $1$.
Please pause for a moment to note that \emph{measurement changes the state}.
This already shows that \emph{states are not the same as measurement}.\footnote{\textbf{Isn't this obvious?}\label{footnote:statesMeasurements}
Mathematically we already know states are not measurements, 
because states are elements of $\mathcal{H}$ (or its tensor
powers) and measurements are certain projections.
But the emphasis is deserved
because it is tempting to relapse into 
a classical world that elides this `type level' distinction.
I can imagine a first lecture on classical physics that
begins `Physics is about physical states, i.e.~what we can measure
\ldots '.}
This is one of the three basic ways in which quantum physics differs from classical physics.
(The other two are superposition and entanglement; all three are critical to quantum computation.) We will see in Theorem~\ref{thm:QDFT}
and \S\ref{subsec:measuringStabiliser} how this `collapse of the wave function on measurement' 
is exploited in quantum computation and in quantum error correction.

After an act of measurement, you learn the result: this is a classical bit~$0$ or $1$,
depending on whether the new quantum state is $\ket{0}$ or $\ket{1}$.
In the computational setting, imagine a wire connecting the quantum computer
to an ordinary classical computer that carries all these measurement results.
As a quick exercise, what does $X$-basis measurement do?\footnote{\textbf{Answer:}
it projects to $\ket{+}$ and $\ket{-}$ and after measurement you learn $+$ if the
new state is $\ket{+}$ and $-$ if the new state is $\ket{-}$.}
What happens if you measure
the qubit $\ket{0}$ in the $X$-basis and then in the $Z$-basis?\footnote{\textbf{Answer:} 
the $X$-basis measurement gives $\ket{+}$ and $\ket{-}$ with equal probability (and you learn which); the $Z$-basis
measurement then gives $\ket{0}$ and $\ket{1}$ again with equal probability (and again you learn which).}
 Well done, now you do not have to read the rest of this section, which is by far the hardest
 part of these notes.  Please skip to \S\ref{sec:twoQubits}.

\subsection{The Stern--Gerlach experiment${}^\star$}
Subsections marked $\star$ should be skipped. Seriously, haven't you been warned enough?
The diagram below show the quantum circuit abstraction of the Stern--Gerlach experiment,
in which 
a spin up qubit (mathematically $\ket{0}$) is put through a splitter (mathematically
the Hadamard gate $H$). Its spin is then measured by a gadget
implementing the $Z$-basis measurement in
Definition~\ref{defn:measureZ}, as shown by the meter.
Since just before measurement the qubit is in the plus state
\smash{$\mfrac{1}{\sqrt{2}} \bigl( \ket{0} +
\ket{1}\bigr)$}, Definition~\ref{defn:measureZ} implies that
the experimental results show an even split between
measurements of $0$ and $1$.
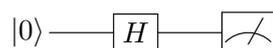
\[ \Qcircuit{& \llap{$\ket{0}$}\  & \gate{H} & \meter{} } \]

Consistent with Definition~\ref{defn:measureZ}, exactly the same statistics are observed if we instead start
with spin down qubits (mathematically $\ket{1}$) which are transformed by the splitter
to the minus state \smash{$\mfrac{1}{\sqrt{2}} \bigl( \ket{0} -
\ket{1}\bigr)$}.
Now suppose that we put the measured qubit through a second splitter, and measure
again, as shown below
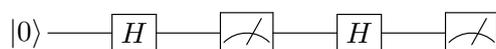
\[ \Qcircuit{& \llap{$\ket{0}$}\  & \gate{H} & \meter{} & \gate{H} & \meter{} } \]
Again it follows
easily from Definition~\ref{defn:measureZ} that both measurements give
an even split between $\ket{0}$ and $\ket{1}$. % half the time and $\ket{1}$ half the time. 
But now suppose 
we do not observe the qubit in between the two splitters, so the relevant circuit is now
 \[ \Qcircuit{& \llap{$\ket{0}$}\  & \gate{H} & \qw  & \gate{H} & \meter{} } \]
What now are the experimental results?
Correct: because $H^2 = I$, we always measure $\ket{0}$. 
The  Stern--Gerlach experiment shows
that  a \emph{superposition}, such as the
plus state \smash{$\mfrac{1}{\sqrt{2}} \bigl( \ket{0} +
\ket{1}\bigr)$} of the unmeasured qubit between the two splitters, is not `secretly' 
either $\ket{0}$ or $\ket{1}$ --- since if so, the final
experimental results would show the same even split between $\ket{0}$
and $\ket{1}$  ---
but a \emph{different physical state}. This contradicts classical physics.
Note we needed only a single qubit and a $2$-dimensional Hilbert space to do all this:
no continuous wave functions or double-slit experiments were required!\footnote{\textbf{Reference:}
see \cite{RudolphQ} for a rigorous introduction to quantum computation, written for the lay reader
that, no coincidence, continues along these lines. For more quantum circuit abstractions of 
important physics experiments, including a fascinating discussion of Wigner's Friend, see
Maria Violaris' videos: \url{https://www.youtube.com/watch?v=TMBK88Mpg5U}.}

\subsection{Measuring spin in any direction${}^\star$}\label{subsec:lifeDifficult}
For this subsection we shall suppose that qubits are particles such as electrons
which have a \emph{spin}, namely an axis in $\R^3$ and a spin magnitude,
either $+\mfrac{1}{2}$ or $-\mfrac{1}{2}$.
In this
case the Stern--Gerlach measuring apparatus is a pair of magnets.
To measure spin about the axis in the
direction of the normalized vector
$\vecn{n} \in \R^3$, you orient the pair
in the direction of $\vecn{n}$, and observe whether the particle
is pulled in the direction $\vecn{n}$ `forwards', indicating spin $+\mfrac{1}{2}$, 
or in the direction $-\vecn{n}$ `backwards', indicating spin $-\mfrac{1}{2}$.\footnote{\textbf{Physical details:} the
Stern--Gerlach experiment is typically performed with silver atoms because 
they are easier to manage in the laboratory than electrons (cathode rays). The reason
silver is an acceptable substitute
is that silver has 
atomic number $47 = 2 + 8 + 8 + 18 + 10 + 1$,
filling up the $1s, 2s, 2p, 3s, 3p, 3d, 4s, 4p, 4d$ shells, leaving one electron
all on its own in the outermost $5s$ shell. (This is an exception to the Aufbau
rule, that predicts $\ldots 4d^{9} 5s^2$: the configuration adopted by silver is lower energy.)
Because all the inner electrons resonate  with each other,
the magnetic properties of silver are dominated by the spin of this single outer electron. 
To measure spin in  direction $\vecn{z}$, 
take a strong top magnet and a weaker bottom magnet, making a magnetic field
oriented in the $\vecn{z}$ direction. Spin $\mfrac{1}{2}$ particles entering the field
align their spin with the $\vecn{z}$ direction, and because the field is inhomogeneous,
particles with spin \smash{$+\mfrac{1}{2}$} are pulled up (`forward', project to $\ket{0}$, measure $0$) and 
particles with spin \smash{$-\mfrac{1}{2}$} are pulled down (`backwards', project to $\ket{1}$, measure $1$).
The deflection is quantized, and is always a multiple by $+1, -1, +2, -2$ etc,
of a minimum deflection; this minimum deflection
is determined by  Planck's constant and the magnetic field strength.
(For silver atoms, there are no excited states and only two deflections are observed, one up and one down.)
This is already not consistent with classical physics, which says that there should exist particles
with spin very weakly aligned in the  $\vecn{z}$ direction
that are deflected up by a very small but non-zero amount.}

\begin{tikzpicture}

\end{tikzpicture}

\subsubsection*{Measurement in a rotated direction}
Measurement in direction $\vecn{z}$ is observed experimentally to result
in `up', $+1$, `forward' on particles in state $\ket{0}$ and
 `down', $-1$, `backward' on particles in state $\ket{1}$.
(We may take this as the \emph{definition} of `spin up' and `spin down', and ignore
that the correct quantized unit is $\hbar/2$.)
Recall that the $Z$-matrix has $1$-eigenvector $\ket{0}$, and $-1$-eigenvector $\ket{1}$ forming the $Z$-basis.
If we rotate the apparatus by a small angle $\theta$ about the $y$-axis, so that it now points in the 
direction $\cos \theta \vecn{z} + \sin \theta \vecn{x}$,
then the $Z$ matrix `rotates' to \marginpar{\raisebox{6pt}{ $Z = \left(\begin{matrix} 1 & 0 \\ 0 & -1 \end{matrix}
\right)$}}
\[ R = \left( \begin{matrix} \cos \theta & \sin \theta \\ \sin \theta & -\cos \theta
\end{matrix} \right) = (\cos \theta) Z + (\sin \theta) X \]
Note that $R$ is still a Hermitian matrix and from the determinant and trace, you can see
it still has eigenvalues $+1$ and $-1$; moreover it shows
 the only
way to evolve $Z$ by a one-parameter subgroup so that when $\theta = \mfrac{\pi}{2}$ it becomes $X$.\footnote{\textbf{Aside:} the
author has thought about variations of this argument for over four years and still
cannot decide whether this is (a) a correct, physically motivated argument
using the isotropy of space, that shows that the Pauli matrices behave like `vectors'
in that measurement in direction $\cos \theta \vecn{z} + \sin\theta \vecn{x}$
corresponds to the Hermitian matrix $\cos \theta Z + \sin \theta X$ \emph{or}
(b) a complete cheat.} The calculations
\begin{align*}
 \left( \begin{matrix} \cos \theta & \sin \theta \\[3pt] \sin \theta & -\cos \theta 
\end{matrix} \right) \left( \begin{matrix} \cos \mfrac{\theta}{2} \\[3pt] \sin \mfrac{\theta}{2}
\end{matrix} \right) &= 
\left(\begin{matrix} \cos \theta \cos \mfrac{\theta}{2} + \sin \theta \sin \mfrac{\theta}{2} \\[3pt]
\sin \theta \cos \mfrac{\theta}{2} - \cos \theta \sin \mfrac{\theta}{2} \end{matrix} \right) 
= \left( \begin{matrix} \cos \mfrac{\theta}{2} \\[3pt] \sin \mfrac{\theta}{2}
\end{matrix}
\right)  \\
 \left( \begin{matrix} \cos \theta & \sin \theta \\[3pt] \sin \theta & -\cos \theta 
\end{matrix} \right) \left( \begin{matrix} \sin \mfrac{\theta}{2} \\[3pt] -\cos \mfrac{\theta}{2}
\end{matrix} \right) &= 
\left(\begin{matrix} \cos \theta \sin \mfrac{\theta}{2} - \sin \theta \cos \mfrac{\theta}{2} \\[3pt]
\sin \theta \sin \mfrac{\theta}{2} + \cos \theta \cos \mfrac{\theta}{2} \end{matrix} \right) 
= \left( \begin{matrix} -\sin \mfrac{\theta}{2} \\[3pt] \cos \mfrac{\theta}{2}
\end{matrix}
\right) 
\end{align*}
show that $R$ has eigenvector $\cos \mfrac{\theta}{2} \ket{0} + \sin \mfrac{\theta}{2} \ket{1}$
with eigenvalue $1$ and eigenvector 
$-\sin \mfrac{\theta}{2} \ket{0} + \cos \mfrac{\theta}{2} \ket{1}$
with eigenvalue $-1$. 
These are eigenstates that are the two possible results of measurement in
the direction $\cos \theta \vecn{z} + \sin \theta \vecn{x}$.
By the Born rule, starting with $\ket{0}$, the probability of measuring
$+1$ `forwards' and $-1$ `backwards' are
\[ \bigl| \cos \mfrac{\theta}{2} \braket{0}{0} + \sin \mfrac{\theta}{2}\braket{ 1}{0} \bigr|^2 
= \cos^2\! \mfrac{\theta}{2}, \quad
\bigl| -\sin \mfrac{\theta}{2} \braket{0}{0} + \cos \mfrac{\theta}{2}\braket{ 1}{0} \bigr|^2 
= \sin^2\! \mfrac{\theta}{2}, \]
respectively. This is both a theoretical prediction and an experimental fact.
%\[ \braket{\cos \mfrac{\theta}{2} + \sin

\subsubsection*{From Stern--Gerlach to $\SU_2 \rightarrow \SO_3(\R)$}
Imagine rotating the apparatus very slowly, and repeatedly performing measurements.
Occasionally, we'll be unlucky and measure `backwards', in which case we start again, but
we can safely assume that all measurements are `forwards'.\footnote{\textbf{Not a lie:} by the Born rule,
the probability
of a `backwards' measurement after a rotation by $\mfrac{\theta}{N}$ 
is \smash{$\sin^2 \!\mfrac{\theta}{N}$ which is at most $\mfrac{\theta^2}{N^2}$} and so
by a union bound the probability of a `backwards' measurement in $N$ steps 
is at most $\theta^2/N$. %\smash{$\mfrac{\theta^2}{N}$}. 
Thus by rotating sufficiently slowly (not forgetting to measure after each tiny rotation)
we can make the chance of a `backwards' measurement --- corresponding in other settings to a 
dissipation of the quantum state --- arbitrarily small.}

\begin{exercise}\label{ex:spin}
Suppose we rotate the apparatus in total by an angle $\theta$. What is the new quantum state?
What unitary operator $U$ corresponds to this evolution of the starting state $\ket{0}$?
What is the conjugate of $Z$ by this~$U$? Go on to discover the double cover $\SU_2
\rightarrow \SO_3(\R)$ and interpret everything so far with Lie algebras.\footnote{\textbf{Answer:} the new quantum state is the
eigenstate $\cos \mfrac{\theta}{2} \ket{0} + \sin \mfrac{\theta}{2} \ket{1}$ \label{footnote:adjointRep}
above
and 
\[ U = \left(\begin{matrix} \cos \mfrac{\theta}{2} & -\sin \mfrac{\theta}{2} 
\\[3pt] \sin \mfrac{\theta}{2} & \cos \mfrac{\theta}{2} \end{matrix} \right).\]
Calculation shows that 
\begin{align*} UZU^{-1} &=
\left(\begin{matrix} \cos \mfrac{\theta}{2} & -\sin \mfrac{\theta}{2} 
\\[3pt] \sin \mfrac{\theta}{2} & \cos \mfrac{\theta}{2} \end{matrix} \right)
\left(\begin{matrix} 1 & 0 \\ 0 & -1\end{matrix}\right)
\left(\begin{matrix} \cos \mfrac{\theta}{2} & \sin \mfrac{\theta}{2} 
\\[3pt] -\sin \mfrac{\theta}{2} & \cos \mfrac{\theta}{2} \end{matrix} \right) \\
&= \left(\begin{matrix} \cos \mfrac{\theta}{2} & \sin \mfrac{\theta}{2} \\[3pt]
\sin \mfrac{\theta}{2} & -\cos \mfrac{\theta}{2} \end{matrix} \right)
\left(\begin{matrix} \cos \mfrac{\theta}{2} & \sin \mfrac{\theta}{2} 
\\[3pt] -\sin \mfrac{\theta}{2} & \cos \mfrac{\theta}{2} \end{matrix} \right)
\\ &= \left( \begin{matrix} \cos^2 \mfrac{\theta}{2} - \sin^2 \mfrac{\theta}{2} &
2 \cos \mfrac{\theta}{2} \sin  \mfrac{\theta}{2} \\[3pt] 2 \cos  \mfrac{\theta}{2}\sin  \mfrac{\theta}{2} &
\sin^2  \mfrac{\theta}{2} - \cos^2  \mfrac{\theta}{2} \end{matrix} \right) \\ 
&= \left(\begin{matrix} \cos \theta & \sin \theta \\ \sin \theta & -\cos \theta \end{matrix}\right) 
\end{align*}
 Thus when a quantum state
 transforms by $U$, physical measurements (i.e.~Hermitian matrices) transform by conjugation by $U$,
 and the calculation above shows that measurement in the $\vecn{z}$ direction transforms
 to measurement in the $\cos \theta \vecn{z} + \sin \theta \vecn{x}$ direction.
 (We agreed this was the correct physical interpretation of the final matrix, which is $R$ from earlier: see footnote~8.)
More mathematically, quantum states transform by the natural representation
of $\SU_2$ and measurements transform by the adjoint representation of $\SU_2$,
acting by conjugacy on Hermitian matrices.

We carry on to do the calculations needed for \S\ref{subsec:Bloch}.
Restricting the adjoint action of $\SU_2$ to
the subspace $\langle X, Y, Z \rangle$ of $2 \times 2$ complex matrices spanned by the Pauli matrices, 
we get the double cover $\SU_2 \rightarrow
\SO_3(\R)$. The image is in the orthogonal group $\SO_3(\R)$ because conjugation preserves the inner product on 
$\langle Z, X, Y \rangle_\mathbb{R}$ defined by $(P, Q) = \mfrac{1}{2}\tr PQ$
having $Z, X, Y$ as an orthonormal~basis.
The more general version of the calculation
above needs~$U$ defined so that $U\! \ket{0} = \alpha \!\ket{0} + \beta\! \ket{1}$, 
for arbitrary $\alpha, \beta \in \C$ with $|\alpha|^2 + |\beta|^2 = 1$. To make $U$ unitary we take
$U = \left( \begin{matrix} \alpha & -\overline{\beta} \\ \beta & \overline{\alpha} 
\end{matrix} \right)$
and then
\[ UZU^{-1} \!=\! \left( \begin{matrix} \alpha & -\overline{\beta} \\ \beta & \overline{\alpha} 
\end{matrix} \right) 
Z \left( \begin{matrix} \overline{\alpha} & \overline{\beta} \\ -\beta & \alpha 
\end{matrix} \right) 
\!=\! \left(\begin{matrix} \alpha & \overline{\beta} \\ \beta & -\overline{\alpha} \end{matrix}\right)
\left( \begin{matrix} \overline{\alpha} & \overline{\beta} \\ -\beta & \alpha 
\end{matrix} \right) 
\!=\! \left( \begin{matrix} |\alpha^2| - |\beta|^2 & 2\alpha \overline{\beta} \\
2\overline{\alpha} \beta & - |\alpha|^2 + |\beta|^2  \end{matrix} \right)
\]
shows that $UZU^{-1} = a_Z Z + a_X X + a_Y Y$ 
where $a_Z = |\alpha|^2 - |\beta|^2$,
$a_X = 2\Re \alpha\overline{\beta}$ and $a_Y = -2\Im \alpha\overline{\beta}$.
Since $Z$ has $\ket{0}$ as its $1$-eigenstate, the conjugate $UZU^{-1}$ 
has  $U\ket{0} =\alpha \!\ket{0} + \beta\! \ket{1}$ as its $1$-eigenstate.
Thus, in the measurement representation, $\vecn{z} \mapsto a_Z \vecn{z} + a_X \vecn{x} + a_Y \vecn{y}$.
Similar calculations conjugating the Pauli $X$ and $Y = iXZ$ matrices show that
\begin{align*} U X U^{-1} &= 
 \left(\begin{matrix}  -\overline{\beta} &\alpha \\  \overline{\alpha} & \beta \end{matrix}\right)
\left( \begin{matrix} \overline{\alpha} & \overline{\beta} \\ -\beta & \alpha 
\end{matrix} \right) 
=
\left( \begin{matrix} -2 \Re \alpha \beta & \alpha^2 - \overline{\beta}^2 \\
\overline{\alpha}^2 - \beta^2 & 2 \Re \alpha \beta \end{matrix} \right) \\
U Y U^{-1} &= 
\left(\begin{matrix}  -i \overline{\beta} & -i\alpha \\  i\overline{\alpha} & -i\beta \end{matrix}\right)
\left( \begin{matrix} \overline{\alpha} & \overline{\beta} \\ -\beta & \alpha 
\end{matrix} \right) 
=
\left( \begin{matrix} -2 \Im \alpha \beta & -i(\alpha^2 + \overline{\beta}^2) \\
i(\alpha^2 + \overline{\beta}^2)  & 2 \Im \alpha \beta \end{matrix} \right) 
 \end{align*}
and so, using the coefficients of $Z$, $X$, $Y$ as the column of the matrix, we get the explicit double-cover map
\[ \left( \begin{matrix} \alpha & -\overline{\beta} \\ \beta & \overline{\alpha} \end{matrix} \right)
\mapsto \left( \begin{matrix} |\alpha|^2 - |\beta|^2  & -2\Re \alpha\beta & -2\Im \alpha \beta \\
2 \Re \alpha\overline{\beta} & \Re (\alpha^2 - \overline{\beta}^2) & \Im (\alpha^2 + \overline{\beta}^2) \\
-2 \Im \alpha\overline{\beta} & - \Im (\alpha^2 - \overline{\beta}^2) & \Re (\alpha^2 + \overline{\beta}^2)
\end{matrix}\right)  \]

A corollary
is the Lie algebra isomorphism $\su_2 \cong \mathsf{so}_3(\R)$.
It is arguably more elegant to bring in the Lie algebra earlier
and restrict the adjoint action
instead to $\langle iZ, iX, iY\rangle$. We then get $\SU_2$ acting by conjugacy
on its Lie algebra $\su_2$ of anti-Hermitian matrices~$M$ such that \smash{$M = -\overline{M}^t$};
an explicit isomorphism $\SU_2(\R) \cong 
(\R^3, \wedge)$ is then defined on this basis by $-\mfrac{iX}{2} \mapsto \vecn{x}$, $-\mfrac{iY}{2}  \mapsto \vecn{y}$,
$-\mfrac{iZ}{2}  \mapsto \vecn{z}$. In turn an isomorphism $(\R^3, \wedge) \cong \so_3(\R)$
is given by mapping $\vecn{n}$ to the generator of the infinitesimal rotation 
with axis $\vecn{n}$. In our chosen basis, 
\[ -\mfrac{iZ}{2} %= \mfrac{1}{2} \left( \begin{matrix} -i & 0 \\ 0 & i \end{matrix} \right)
\mapsto \left( \begin{matrix} 0 & \cdot & \cdot \\ 
\cdot & 0 & \;\llap{$-$}1 \\ \cdot & 1 & 0 \end{matrix} \right), \
-\mfrac{iX}{2} %= \mfrac{1}{2} \left( \begin{matrix} -i & 0 \\ 0 & i \end{matrix} \right)
\mapsto \left( \begin{matrix} 0\!\! & \cdot & 1 \\ 
\cdot\!\! & 0 & \cdot \\  -1\!\!  & \cdot & 0 \end{matrix} \right) , \
-\mfrac{iY}{2} %= \mfrac{1}{2} \left( \begin{matrix} -i & 0 \\ 0 & i \end{matrix} \right)
\mapsto \left( \begin{matrix} 0 & \;\llap{$-$}1 & \cdot \\ 
\ 1 & 0 & \cdot \\ \cdot & \cdot & 0 \end{matrix} \right) 
\]
See \S\ref{subsec:whyUnitary} for why anti-Hermitian matrices are natural in this context.
}

\end{exercise}

The previous exercise is, as far as the author has been able to understand,
what physicists mean when they say that the symmetry group
of a qubit is $\SU_2$ or, rephrased,
that a qubit `is' a vector in the standard representation of~$\SU_2$.
At the next level up there is the standard representation $\mathcal{K}$ of~$\mathrm{SU}_3$,
in  which the canonical basis vectors label the up, down and strange flavours of quarks and general
elements represent superpositions of these flavour states. 
The $27$-dimensional tensor product $\mathcal{K} \hskip1pt\otimes\hskip1pt \mathcal{K} \hskip1pt\otimes\hskip1pt \mathcal{K}$
decomposes into irreducible subrepresentations that partition baryons made from the up, down and strange quarks.
For instance the proton and neutron `live' in the same irreducible $8$-dimensional
representation: this is the famous Gell-Mann eightfold way.
Be warned that if you search on the web, you are very likely
to find this, rather than the simpler $\SU_2$ symmetry group relevant to a spinor qubit.

\begin{exercise}\label{ex:fullTurn}
Suppose that the measurement apparatus is rotated slowly by a full turn, about the $\vecn{y}$-axis as usual, measuring
as usual after each small rotation. What is the new quantum state? Interpret
this as paths in $\SU_2$ and $\SO_3(\R)$.\footnote{\textbf{Answer:}
the final state is not the starting state $\ket{0}$ but instead $-\ket{0}$, which differs
by an unobservable global phase. 
Particles with this property are called \emph{spinors}.
Thus in $\SO_3(\R)$ we made a full circuit (visiting all rotations about the $\vecn{y}$-axis 
by angles
between $0$ and $2\pi$) but in $\SU_2$ we travelled
only from $I$ to the antipodal point $-I$. This is consistent with the double cover
because $U$ and $-U$ have the same image in $\SO_3(\R)$; in fact the kernel
of the double cover homomorphism is the subgroup $\{I, -I\}$. While the global
phase in $-\ket{0}$ is unobservable, if the rotated qubit is entangled in a Bell pair (see \S\ref{sec:twoQubits}),
the phase difference between rotated and unrotated qubits can be measured experimentally.%
%More generally, a trip by an angle $\theta$ in $\SU_2$ corresponds to a trip
%by an angle $2\theta$ in $\SO_3(\mathbb{R})$; this explains why in the Bloch sphere
%representation of normalized states (up to phase), $\ket{0}$ and $\ket{1}$ are antipodal
%rather than orthogonal. For more on the Bloch sphere see \cite[page~15]{NielsenChuang}.
} 
\end{exercise}

\subsection{Summary: the Bloch sphere${}^\star$}\label{subsec:Bloch}
Generalizing Exercise~\ref{ex:fullTurn}, it follows from the explicit double cover map in footnote~\ref{footnote:adjointRep}
that the image of the one-parameter
subgroup 
\[ \theta \mapsto \left( \begin{matrix} \cos \theta &  -\sin \theta \\ \sin \theta & \cos \theta
\end{matrix}\right) \in \SU_2 \]
under the double cover map is the one-parameter subgroup
\[ \theta \mapsto \left( \begin{matrix} \cos 2\theta & \sin2\theta & 0 \\ -\sin2\theta & \cos 2\theta & 0 \\
0 & 0 & 1 \end{matrix} \right) \in \SO_3(\mathbb{R}). \]
For example, as seen in Exercise~\ref{ex:fullTurn},
a trip along the path defined by $0 \le \theta \le \pi$ in $\SU_2(\R)$, which overall flips the phase
of the qubit, becomes a closed loop in $\SO_3(\R)$; this corresponds to the unobservability
of global phase.
Going the other way, by varying the axis of rotation in $\SO_3(\R)$ (the spin axis that we can measure)
we can trace out any path we like in $\SU_2$, or equivalently, move the starting qubit $\ket{0}$ to
an arbitrary $\alpha\! \ket{0} + \beta\! \ket{1}$. 
Since elements of $\SO_3(\R)$ are rotations, acting on the $2$-sphere, this implies that we should
be able to visualize states (up to phase) as points on the $2$-sphere. 
This is the \emph{Bloch sphere} model, shown in Figure~\ref{fig:Bloch}.

\begin{figure}[h!]
\begin{center}
\begin{tikzpicture}[scale=3]

    \def\R{1}       
    \def\tilt{0.4}

    \draw (0,0) circle (\R);

    \draw[dashed, gray] (\R,0) arc (0:180:\R cm and \tilt cm);
    \draw (\R,0) arc (0:-180:\R cm and \tilt cm);

    \draw[->] (0,0) -- (0, \R + 0.3) node[above] {$z$};
    \draw (0,0) -- (0, -\R - 0.3);

    \draw[->] (0,0) -- (\R + 0.3, 0) node[right] {$y$};
    \draw (0,0) -- (-\R - 0.3, 0);

    \draw[->, thick] (0,0) -- (-135:\R + 0.4) node[below] {$x$};
    \draw[dashed] (0,0) -- (45:\R + 0.4);

    \fill (0,\R) circle (0.7pt) node[above left] {$\ket{0}$};
    
    \fill (0,-\R) circle (0.7pt) node[below left] {$\ket{1}$};
    
    \fill (\R,0) circle (0.7pt) node[above right] {$\mfrac{1}{\sqrt{2}}\bigl(\ket{0} + i \ket{1}\bigr)$};
    
    \fill (-\R,0) circle (0.7pt) node[above left] {$\mfrac{1}{\sqrt{2}}\bigl(\ket{0} - i \ket{1}\bigr)$};

    \fill (0,0) circle (0.5pt);
    
    \node[below] at (-135:\R) {$\ket{+}$};
    \node[above] at (45:\R) {$\ket{-}$};

\end{tikzpicture}
\end{center}
\caption{The Bloch sphere showing the $Z$-basis $\ket{0}$, $\ket{1}$, the $X$-basis $\ket{+}$, $\ket{-}$ and 
the normalized eigenvectors of $Y = iXZ$.\label{fig:Bloch}}
\end{figure}
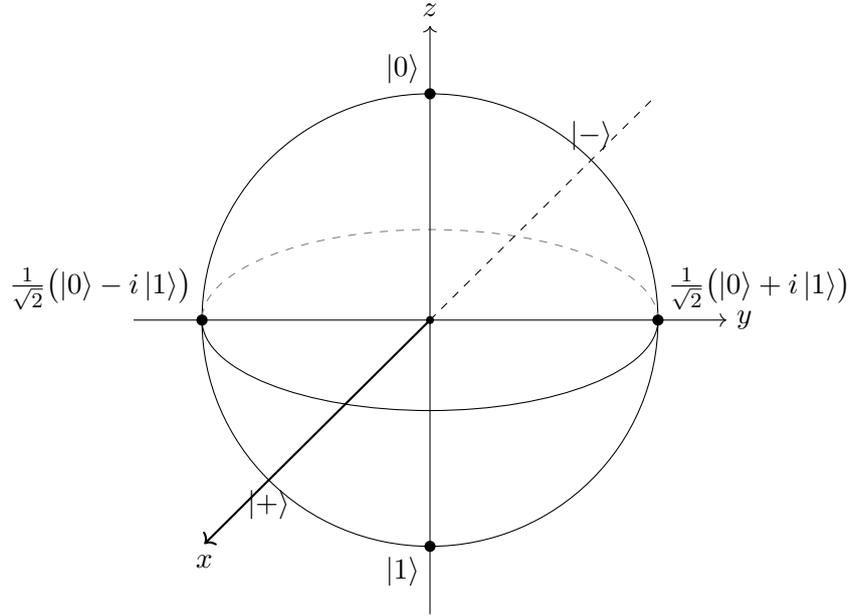

\begin{exercise}
Check that a normalized qubit, up to phase, has two free parameters,
and so the $2$-dimensional Bloch sphere has the right dimension.\footnote{\textbf{Solution:} 
a general normalized qubit is $\alpha\! \ket{0} + \beta\!\ket{1}$ where $|\alpha|^2 + |\beta|^2 = 1$.
Working up to phase we can multiply through by $e^{-i \theta}$ where $\theta$ is the argument
of $\alpha$, to reduce to the case where $\alpha$ is real. So the state is
$c\ket{0} + (d + i e)\ket{1}$ where
 $c^2 + d^2 + e^2 = 1$, matching points on the $2$-sphere.}
 \end{exercise}
 
By the angle doubling remark above, \emph{orthogonal} quantum states, such as the $Z$-eigenvectors 
$\ket{0}$ and $\ket{1}$
correspond to \emph{antipodal} points on the Bloch sphere.

\begin{exercise}\label{ex:Y}
Check that the Pauli $Y$ matrix \marginpar{\qquad\quad$\left(\begin{matrix} 0 & -i \\ i & 0 \end{matrix} \right)$}
defined by $Y=iXZ$, as shown in 
the margin, has eigenvectors as shown on the $y$-axis of the Bloch sphere.\footnote{\textbf{Solution:}
after multiplying through by the normalization factor $\mfrac{1}{\sqrt{2}}$ we have
\[ Y \bigl(\ket{0} \pm i \ket{1}\bigr) = \left( \begin{matrix}  0 & -i \\ i & 0 \end{matrix} \right)
\left( \begin{matrix} 1 \\ \pm i \end{matrix} \right) = 
\left( \begin{matrix} \pm 1 \\  i \end{matrix} \right) = \pm \left( \begin{matrix} 1 \\ \pm i \end{matrix} \right)
 = \pm \bigl(\ket{0} \pm i \ket{1}\bigr). \]
%now normalize as we do for $\ket{+} = \mfrac{1}{\sqrt{2}} \bigl( \ket{0} + \ket{1} \bigr)$.
\vspace*{-12pt}
}
\end{exercise}

For our explicit double cover map,
the map from normalized qubits to points on the Bloch sphere is
\begin{align*}
\alpha\! \ket{0} &{}+ \beta\!\ket{1} = \left( \begin{matrix}
\alpha & -\overline{\beta} \\ \beta & \overline{\alpha} \end{matrix} \right)
\ket{0} 
\\ &\!\!\mapsto \left( \begin{matrix} |\alpha|^2 - |\beta|^2  & -2\Re \alpha\beta & -2\Im \alpha \beta \\
2 \Re \alpha\overline{\beta} & \Re (\alpha^2 - \overline{\beta}^2) & \Im (\alpha^2 + \overline{\beta}^2) \\
-2 \Im \alpha\overline{\beta} & - \Im (\alpha^2 - \overline{\beta}^2) & \Re (\alpha^2 + \overline{\beta}^2)
\end{matrix}\right) \! \left(\begin{matrix} 1 \\ 0 \\ 0
\end{matrix} \right) \!=\! \left( \begin{matrix} |\alpha|^2 - |\beta|^2 \\ 2\Re \alpha\overline{\beta} \\ 
-2\Im \alpha\overline{\beta} \end{matrix} \right)\! .
\end{align*}
Note that the axes are ordered $(z,x,y)$.
This is a reformulation of the remark in footnote~\ref{footnote:adjointRep} that in the measurement representation
$\vecn{z} \mapsto a_Z \vecn{z} + a_X \vecn{x} + a_Y \vecn{y}$,
where $a_Z = |\alpha|^2 -|\beta|^2$, $a_X = 2\Re \alpha\overline{\beta}$, $a_Y = -2\Im \alpha\overline{\beta}$.
 In the same footnote, we saw that
since $Z$
has eigenbasis $\ket{0}$, $\ket{1}$, the conjugate
\begin{align*} UZU^{-1} &=
%\left( \begin{matrix} |\alpha^2| - |\beta|^2 & 2\alpha \overline{\beta} \\
%2\overline{\alpha} \beta & - |\alpha|^2 + |\beta|^2  \end{matrix} \right)
 \bigl(|\alpha|^2 - |\beta|^2 \bigr)Z  +  2 \Re \alpha\overline{\beta} X - 2\Im \alpha\overline{\beta}  Y \\
 &= a_Z Z + a_X X + a_Y Y
\end{align*}
has $\alpha \!\ket{0} + \beta \! \ket{1}$ as a $1$-eigenvector. (And, from the second 
column of $U$, $-\overline{\beta}\ket{0} + \overline{\alpha}\ket{1}$
spans its orthogonal complement in $\HH$ and is a $-1$-eigenvector.) Thus qubits transform by $\SU_2$
whereas physical measurements --- thought of either as Hermitian matrices as immediately
above, or directions of our measuring apparatus in~$\R^3$ --- transform by the image $\SO_3(\R)$
of $\SU_2$ under the double cover map; this is the adjoint representation of $\SU_2$.
In particular by taking \smash{$\alpha = \mfrac{1}{\sqrt{2}}$, $\beta = \mfrac{i}{\sqrt{2}}$} 
and noting that \smash{$a_Z = 0$, $a_X = 0$ and $a_Y = \mfrac{1}{2}$}, 
we get an alternative solution to Exercise~\ref{ex:Y}.
For more on the Bloch sphere see \cite[page~15]{NielsenChuang}.

\section{Two qubits and the copy rules}\label{sec:twoQubits}

\subsection{Two qubits}\label{subsec:twoQubits}
The correct way to model two qubits is by the tensor product $\HH \otimes \HH$
where, as always in these notes, the tensor product is over~$\C$.
Experience shows that it works very well to take as the \emph{definition} of $\HH \otimes \HH$ that it is the
vector space spanned by the symbols $\ket{0} \otimes \ket{0}$, $\ket{0} \otimes \ket{1}$,
$\ket{1} \otimes \ket{0}$, $\ket{1} \otimes \ket{1}$, which we quickly rewrite as 
either $\ket{0}\!\ket{0}$, $\ket{0}\!\ket{1}$, $\ket{1}\!\ket{0}$, $\ket{1}\!\ket{1}$,
or, particularly when more than two qubits (see \S\ref{subsec:Zbasis}) are involved, as
$\ket{00}$, $\ket{01}$, $\ket{10}$, $\ket{11}$. This is the $Z$-\emph{basis} of $\HH \otimes \HH$
of eigenvectors for $Z \otimes Z$.\footnote{\textbf{Maths:} of course if you prefer
a more high-brow definition of the tensor product, then you are very welcome to it. It might
seem that omitting the tensor product sign in $\ket{0}\!\ket{1}$ would create an ambiguity
with symmetric powers, but, unlike most things, this has never confused the author.}

\subsubsection*{CNOT gates}
By far the most important $2$-qubit gate is the CNOT gate. The inputs to the CNOT
gate are a \emph{control} qubit and a \emph{target qubit}. Its matrix in the
basis $\ket{0}\!\ket{0}$, $\ket{0}\!\ket{1}$, $\ket{1}\!\ket{0}$, $\ket{1}\!\ket{1}$, supposing 
that the first qubit is the control qubit and the second qubit is the target qubit, is 
\begin{align*} & \qquad\! \scriptstyle \ket{00} \ \ket{01}\ \ket{10} \ \ket{11} \\[-6pt]
\CNOT &= \left(\setlength{\arraycolsep}{6pt} \begin{matrix}1 & 0 & 0 & 0 \\ 0 & 1 & 0 & 0 \\ 0 & 0 & 0 & 1 \\ 0 & 0 & 1 & 0  \end{matrix}\right)
\qquad \raisebox{15pt}{\Qcircuit @C=12pt @R=18pt {
 & \qw  & \ctrl{1} & \qw & \qw  \\ & \qw & \targ{} & \qw& \qw
 }}
 \end{align*}
The circuit diagram for CNOT is shown right above, with the usual convention that the first qubit is on the top wire.
Given any $\ket{\psi} = \alpha \!\ket{0} + \beta \!\ket{1} \in \HH$, we have
\begin{align*}
\CNOT \ket{0} \!\ket{\psi} &= \CNOT \bigl( \alpha\! \ket{0} \!\ket{0} + \beta\! \ket{0} \!\ket{1} \bigr)
= \alpha\! \ket{0} \!\ket{0} + \beta\! \ket{0} \!\ket{1} = \ket{0}\!\ket{\psi}, \\[-3pt]
\CNOT \ket{1}\! \ket{\psi} &= \CNOT \bigl( \alpha\! \ket{1} \!\ket{0} + \beta\! \ket{1} \!\ket{1} \bigr)
= \alpha\! \ket{1}\! \ket{1} + \beta\! \ket{1} \!\ket{0} = \ket{1}\! X\ket{\psi}. \end{align*}
Thus one often says that CNOT `flips the target qubit if the control qubit is set'.\footnote{\textbf{Misleading:}
in practice the control qubit is very often not $\ket{0}$ or $\ket{1}$ but
instead the superposition $\ket{+} = \mfrac{1}{\sqrt{2}} \bigl( \ket{0} + \ket{1} \bigr)$, and so this `classical' account of what CNOT does makes no sense. Still it seems to work for everyone.} In the setup of Definition~\ref{defn:controlledGate}, CNOT is a controlled NOT gate.

\begin{exercise}\label{ex:makeBell}
Show that \smash{$\CNOT \ket{+}\ket{0} = \mfrac{1}{\sqrt{2}} \bigl( \ket{00} + \ket{11} \bigr)$} as shown
diagramatically below.\vspace*{2pt}
\[ \Qcircuit @C=12pt @R=18pt {
    \llap{$\ket{+}$}\; & \qw  & \ctrl{1} & \qw & \qw  \\ 
    \llap{$\ket{0}$}\; & \qw & \targ{} & \qw& \qw 
    } \quad\raisebox{-12pt}{$\mfrac{1}{\sqrt{2}} \bigl( \ket{00} + \ket{11} \bigr)$} \]

\smallskip
\noindent The output is known as the \emph{Bell state}. Is there a meaningful way to label the 
two output wires separately, as we did for the input wires?\footnote{\textbf{Solution:} by linearity,\label{footnote:CNOTlinear}
\[ \CNOT \ket{+}\ket{0} \!=\! \CNOT \mfrac{1}{\sqrt{2}}\bigl( \ket{0} + \ket{1}\bigr) \ket{0} \!=\!
 \mfrac{1}{\sqrt{2}} \CNOT \ket{00} + \mfrac{1}{\sqrt{2}} \CNOT \ket{10} \!=\!
\mfrac{1}{\sqrt{2}}  \ket{00} + \mfrac{1}{\sqrt{2}} \ket{11}\] 
as required. No: since the Bell state
does not factor as a tensor product $\ket{\phi} \otimes \ket{\psi}$ the wires
cannot be labelled separately. In fact the Bell state is in a precise sense,
maximally entangled. See \S\ref{subsec:superdense} for an application of this entanglement.}
\end{exercise}

\subsection{The copy rules}\label{subsec:copyRules}
Here are the diagrammatic copy rules for conjugating faults past CNOT gates. The equivalent
algebra is below.
\begin{align*} 
  &\hspace*{0.1in} \Qcircuit @C=10pt @R=18pt {
   &  \qw  & \rX & & \ctrl{1}  & \qw \\
   &  \qw  & \qw & \qw & \targ{}  & \qw 
 }  \qquad\raisebox{-12pt}{\hspace*{-1.4em}$=$}\  \Qcircuit@C=10pt @R=18pt { &  \qw   & \ctrl{1}  & \qw &  \rX & & \qw \\
   &  \qw  & \targ{}  & \qw & \rX & & \qw}
  \hspace*{0.5in}
   \Qcircuit @C=10pt @R=18pt {
   &  \qw  & \qw & \qw & \ctrl{1}  & \qw \\
   &  \qw  & \rX & &  \targ{}  & \qw 
 } \qquad\raisebox{-12pt}{\hspace*{-1.4em}$=$}\    \Qcircuit @C=10pt @R=18pt {
   &  \qw & \ctrl{1}  & \qw & \qw & \qw & \qw \\
   &  \qw & \targ{}  & \qw & \rX & & \qw
 } 
 \\[3pt]
   &  \CNOT (X \otimes I) = (X \otimes X) \CNOT \qquad    \CNOT (I \otimes X) = (I \otimes X) \CNOT \\[6pt]
  &\hspace*{0.1in} 
 \Qcircuit @C=10pt @R=18pt {
   &  \qw  & \rZ &  & \ctrl{1}  & \qw  \\
   &  \qw  & \qw  & \qw &  \targ{}  & \qw 
 }\qquad\raisebox{-12pt}{\hspace*{-1.4em}$=$}\  \Qcircuit@C=10pt @R=18pt {
   &  \qw  & \ctrl{1}  & \qw & \rZ & & \qw \\
   &  \qw  &  \targ{}  & \qw & \qw & \qw & \qw
 }     \hspace*{0.5in}
  \Qcircuit @C=10pt @R=18pt {
   &  \qw  & \qw & \qw & \ctrl{1}  & \qw \\
   &  \qw  & \rZ &  & \targ{}  & \qw
 } \qquad\raisebox{-12pt}{\hspace*{-1.4em}$=$}\    \Qcircuit @C=10pt @R=18pt {
   &  \qw & \ctrl{1}  & \qw & \rZ & & \qw \\
   &  \qw & \targ{}  & \qw & \rZ & & \qw
 }
   \\[3pt]
   &  \,\CNOT (Z \otimes I) = (Z \otimes I) \CNOT \qquad\ \   \CNOT (I \otimes Z) = (Z \otimes Z) \CNOT  
\end{align*}

\begin{exercise}
Prove the copy rules.\footnote{\textbf{Not a solution:} sorry, this really is something everyone
should do once in their life. If you find the algebra isn't working for you, please check that you 
are composing maps from right-to-left; note this is the opposite direction to circuit diagrams, but consistent
with matrices acting on column vectors and usual mathematical practice. (Although
with a double dose of potential confusion, this might not matter because CNOT is self-inverse.) Another 
way to go wrong is to be inconsistent in the order of qubits.  For instance the matrices for $X \otimes I$ and $I \otimes X$ are, with our usual
order $\ket{00}$, $\ket{01}$, $\ket{10}$, $\ket{11}$ for the basis,
\[ \left( \begin{matrix} 0 & \cdot & 1 & \cdot \\ \cdot & 0 & \cdot & 1 \\ 1 & \cdot & 0 & \cdot \\
\cdot & 1 & \cdot & 0 \end{matrix} \right) \qquad
\left( \begin{matrix} 0 & 1 & \cdot & \cdot  \\ 1 & 0 & \cdot & \cdot  \\ \cdot & \cdot & 0 & 1 \\
\cdot & \cdot  & 1 & 0 \end{matrix} \right).
\]
As a visual guide, $\cdot$ denotes a $0$ from the tensor product factorization, which
puts a $0$ in the positions corresponding to the zeros of $X$; these are marked $\cdot$ in the margin.}
\end{exercise}

In the context of quantum error correction, we often imagine the incoming $X$- or $Z$- as an $X$-fault or $Z$-fault
thrown by a quantum glitch somewhere in the circuit. The copy rules then become rules for fault propagation.

\begin{remark}
From the classical point of view that CNOT is an operation
performed on the target qubit, it may seem very unintuitive
that a $Z$-fault on the target qubit should copy up to
the control qubit: well, that's entanglement for you.
This becomes a recurring theme in quantum error correction:
faults on ancilla states copy up to the target state.
See Exercise~\ref{ex:copyUp} for an example of this.
\end{remark}

Any circuit involving only CNOT and Hadamard gates (and a few other gates we haven't  defined)
\marginpar{\raisebox{-2.75in}{\qquad$X = \left(\begin{matrix}
\cdot & 1 \\ 1 &\cdot \end{matrix}\right)$}}
is characterized by its conjugation action on $X$- and $Z$-faults.\footnote{\textbf{Sketch proof:} all these gates lie in the Clifford
group and so are determined, up to phase, by their conjugacy action on tensor products of the Pauli $X$ and $Z$ operators.
But $H$ and $\CNOT$ do not introduce any phases, so this `up to' is irrelevant. For example, CNOT is the unique phaseless $2$ qubit gate satisfying the copy rules above, which restated as conjugations are
$\CNOT (X\otimes I) \CNOT^{-1} = X \otimes X$, and so on.$\;\Box$
The characterization extends to circuits using
the single qubit $X$, $Y$ and $Z$ gates if one allows for a global phase of $\pm 1$, and to general
Clifford circuits, including the $S$-gate,
allowing for $\pm 1$ and $\pm i$.}
That is, \emph{how the circuit propagates incoming $X$- and $Z$-faults determines it completely}.
This makes the copy rules enormously powerful.

\begin{exercise}
The SWAP gate is
defined on $\HH \otimes \HH$ in a basis independent way by
$\SWAP \ket{\phi}\ket{\psi} = \ket{\psi}\ket{\phi}$. Its $Z$-basis matrix
is shown below, with a circuit implementing SWAP to the right.
\[ \left(\begin{matrix} 1 & 0 & 0 & 0 \\ 0 & 0 & 1 & 0 \\ 0 & 1 & 0 & 0 \\
0 & 0 & 0 & 1 \end{matrix} \right) \qquad
\raisebox{15pt}{\Qcircuit{ &  & \ctrl{1} & \targ{} & \ctrl{1} & \qw \\ &  & \targ{} & \ctrl{-1} & \targ{} & \qw  }}
\]
% This really is the right basis independent form:
% (a \ket{0} + b \ket{1}) \otimes (c \ket{0} + d \ket{1}) = ac \ket{00} + ad \ket{01} + bc \ket{10} + \bd \ket{11}
% -> ac\ket{00} + ad \ket{10} + bc \ket{01} + bd\ket{11} = (c \ket{0} + d\ket{1})(a \ket{0} + b\ket{1}) since we
% can see that ac = ca is the coefficient of \ket{00}, ad = da is the coefficient of \ket{10}, etc.
Prove this by pushing faults through the three CNOT gates using the copy rules.
Equivalently, using indices to denote control and targets in CNOTs, prove
that $\SWAP = \CNOT_{12}\CNOT_{21}\CNOT_{12}$.\footnote{\textbf{Solution}: 
the diagrams below shows that $\SWAP (X \otimes I) = (I \otimes X)\SWAP$
and $\SWAP (Z \otimes I) = (I \otimes Z)\SWAP$.
\[\hspace*{-0.0in}\scalebox{0.9}{$ \Qcircuit @C=10pt @R=18pt {
   &  \qw  & X & & \ctrl{1}  & \qw & \rX & &  \targ{} & \qw & I & & \ctrl{1} & \qw & \qw & \qw & \qw \\
   &  \qw  & \qw & \qw & \targ{}  & \qw & \rX & &  \ctrl{-1} & \qw & \rX & & \targ{} & \qw & \rX & & \qw
 }  \hspace*{0.4in}
\Qcircuit @C=10pt @R=18pt {
   &  \qw  & \rZ & & \ctrl{1}  & \qw & \rZ & &  \targ{} & \qw & \rZ & & \ctrl{1} & \qw & I &  & \qw \\
   &  \qw  & \qw & \qw & \targ{}  & \qw & \qw & \qw &  \ctrl{-1} & \qw & \rZ & & \targ{} & \qw & \rZ & & \qw
 }  $}    
 \]
 
\smallskip 
\noindent Here $I$ is used to denote the place where two $X$- or two $Z$-faults cancel. 
Very similarly, or algebraically by composing on the left and right with SWAP, we have
 $(X \otimes I)\SWAP  = \SWAP (I \otimes X)$
 and  $(Z \otimes I)\SWAP  = \SWAP (I \otimes Z)$.
 By the Clifford characterization, since the putative SWAP gate does the correct thing
 on $X$- and $Z$-faults, and is phaseless, it is SWAP.
 It is interesting to compare
 this circuit with the way to swap two variables in a classical computer without 
 using extra scratch space: in C, using the XOR operator \texttt{\caret}, this is
 \texttt{x  = x\caret y; y = x\caret y; x = x\caret y;}}
  %verb not working in footnote. Indeed.
\end{exercise}

\begin{exercise}\label{ex:makeBellPushing}
Use fault pushing to give an alternative solution to Exercise~\ref{ex:makeBell}.
[\emph{Hint:} the stabiliser group of $\ket{+}\ket{0}$ is generated by $X \otimes I$
and $I \otimes Z$. Push these faults through to get the stabiliser group of the output state.]\footnote{\textbf{Solution:} the
`non-pass-through'
copy rules are precisely the statements that $\CNOT (X \otimes I) =  (X\otimes X)\CNOT$ and
$\CNOT (I \otimes Z) = (Z\otimes Z) \CNOT$. Hence, writing $\ket{\Phi}$ for the output of the circuit
in Exercise~\ref{ex:makeBell},
\begin{align*} \ket{\Phi} = \CNOT \ket{+}\ket{0} &= \CNOT (X\otimes I) \ket{+}\ket{0} = (X \otimes X) \CNOT \ket{+} \ket{0}
= (X \otimes X) \ket{\Phi} \\
\ket{\Phi} =  \CNOT \ket{+}\ket{0} &= \CNOT (I\otimes Z) \ket{+}\ket{0} = (I \otimes Z) \CNOT \ket{+} \ket{0}
= (Z \otimes Z) \ket{\Phi} \end{align*}
showing that $\ket{\Phi}$ has stabiliser group containing $X \otimes X$ and $Z \otimes Z$.
By reversibility, this is the full stabiliser group. The unique state with this stabiliser
group is the Bell state, so the state made by the circuit is $\ket{\Phi}$. In this case the fault pushing solution
is more fuss than a direct calculation, but in the context of the larger ancilla states needed
in quantum error correction, the method of stabiliser subgroups becomes very powerful.
}
\end{exercise}

% Note we don't need measurement on a subset until QEC, even for the Bell measurements can do on both

\subsection{Superdense encoding: entanglement as a resource${}^\star$}\label{subsec:superdense}
Dual to the Bell state preparation circuit in Exercise~\ref{ex:makeBell} 
there is the measurement circuit below.
\[ \Qcircuit{ & \ctrl{1} &  \gate{H} & \meter \\ & \targ{}  & \qw & \meter }
\]
Here $H$ denotes the Hadamard gate and, as usual, meters denote $Z$-basis measurement on each qubit.
For this subsection, the simple single qubit definition of measurement 
(Definition~\ref{defn:measureZ}) can be applied to each qubit in turn.\footnote{\textbf{More elegant:}
the author would rather think of the top line as measurement in the $X$-basis: compare 
the two circuits at the end of \S\ref{subsec:QDFT}. This makes the duality more obvious: preparation of 
$\ket{+}$ dualizes to measurement in the $X$-basis; preparation of $\ket{0}$ dualizes to measurement in the $Z$-basis.
In the ZX-calculus \cite{WetteringZX}, the symmetry is complete: the diagrams for Bell state preparation
and Bell measurement are shown far left and right below, and their fused versions in the middle. 
By a further simplification these become the `cup' and `cap' operators from string-diagram calculus.
None of this matters, but you might enjoy contemplating it, and possibly even being paid to do so.

\medskip
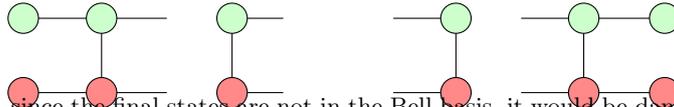
\begin{center}
  \begin{ZX}
  	\zxZ[a=A1]{\phantom{0}} \arrow[from=A1, to=A2] & [9pt] 
						& \zxZ[a=A2]{\phantom{0}} \arrow[from=A2, to=B2] \arrow[from=A2, to=A3] & [9pt]
							 & \node[a=A3] {}; 
							 \\[12pt]
    \zxX[a=B1]{\phantom{0}} \arrow[from=B1, to=B2] & [9pt] & \zxX[a=B2]{\phantom{0}} &  [9pt]\arrow[from=B2, to=B3] 
    						 &  \node[a=B3] {}; 
  \end{ZX}\hspace*{0.2in}
  \begin{ZX} 
  \zxZ[a=A1]{\phantom{0}} \arrow[from=A1, to=B1] \arrow[from=A1, to=A2] & [9pt] \node[a=A2] {}; \\[12pt]
    \zxX[a=B1]{\phantom{0}} \arrow[from=B1, to=B2] & [9pt] \node[a=B2] {};
\end{ZX}\hspace*{0.5in}
 \begin{ZX} 
  \node[a=A1] {}; \arrow[from=A1, to=A2] & [9pt] & \zxZ[a=A2]{\phantom{0}} \arrow[from=A2, to=B2] \\[12pt]
  \node[a=B1] {}; \arrow[from=B1, to=B2] & [9pt] & \zxX[a=B2]{\phantom{0}}
\end{ZX}\hspace*{0.2in}
  \begin{ZX}
  	\node[a=A1] {}; \arrow[from=A1, to=A2] & [9pt] 
						& \zxZ[a=A2]{\phantom{0}} \arrow[from=A2, to=B2] \arrow[from=A2, to=A3] & [9pt]
							 & \zxZ[a=A3]{\phantom{0}} 
							 \\[12pt]
    \node[a=B1] {}; \arrow[from=B1, to=B2] & [9pt] & \zxX[a=B2]{\phantom{0}} &  [9pt]\arrow[from=B2, to=B3] 
    						 & \zxX[a=B3]{\phantom{0}}
  \end{ZX}
\end{center}
}
Let \smash{$\ket{\Phi} = \mfrac{1}{\sqrt{2}} \bigl( \ket{00} + \ket{11} \bigr)$} denote the Bell state.
Since $(X \otimes I)\! \ket{\Phi} = (I \otimes X)\!\ket{\Phi}$ 
we can write $X\! \ket{\Phi}$ without ambiguity to mean `there is an $X$-error on the Bell state'.
The table below showing the measurement results for both top and bottom measurements
(using Definition~\ref{defn:measureZ}
to define measurement on a single qubit) on the four Bell basis states.

\vspace*{-3pt}
\begin{center}
\begin{tabular}{cccccc} \\ \toprule
State & In $Z$-basis & After CNOT & After $H$ & Top & Bottom \\ \midrule
$\ket{\Phi}$ & $\mfrac{1}{\sqrt{2}} \bigl( \ket{00} + \ket{11})$ & $\ket{+}\!\ket{0}$ & $\ket{0}\!\ket{0}$ &  $0$ & $0$ \\[4pt]
$X\!\ket{\Phi}$ & $\mfrac{1}{\sqrt{2}} \bigl( \ket{01} + \ket{10})$ & $\ket{+}\!\ket{1}$& $\ket{0}\!\ket{1}$ &  $0$ & $1$ \\[4pt]
$Z\!\ket{\Phi}$ & $\mfrac{1}{\sqrt{2}} \bigl( \ket{00} - \ket{11})$ & $\ket{-}\!\ket{0}$ & $\ket{1}\!\ket{0}$&  $1$ & $0$ \\[4pt]
$XZ\!\ket{\Phi}$ & $\mfrac{1}{\sqrt{2}} \bigl( \ket{01} - \ket{10})$ & $\ket{-}\!\ket{1}$ & $\ket{1}\!\ket{1}$&  $1$ & $1$ 
\\ \bottomrule 
\end{tabular}
\end{center}

\smallskip
Again let us observe the traditional pause to remember that \emph{measurement changes the state}.
For example, since $(H \otimes I) \CNOT \ket{10} = H \ket{11} = \ket{-}\ket{1}$ \smash{$= \mfrac{1}{\sqrt{2}}
\bigl( \ket{01} - \ket{11} \bigr)$}, measuring
$\ket{11}$ by the Bell basis gadget reports $(0,1)$ and $(1,1)$ with equal probability;
the final states are $\ket{01}$ and $\ket{11}$ respectively.\footnote{\textbf{Pitfall:}\label{footnote:pitfall}
since the final states are not in the Bell basis, it would be dangerously loose to describe
our gadget as `measuring in the Bell basis'; instead this is done by
projecting directly to this basis. In our example,
\[ \ket{10} = \mfrac{1}{2} \bigl( \ket{01} + \ket{10}\bigr) - \mfrac{1}{2} \bigl( \ket{01}  -\ket{10} \bigr) =
\mfrac{1}{\sqrt{2}} X \ket{\Phi} + \mfrac{1}{\sqrt{2}} XZ \ket{\Phi} \]
and so \smash{$\braket{\beta X}{10} = \frac{1}{\sqrt{2}}$}, 
\smash{$\braket{\beta XZ}{10} = \frac{1}{\sqrt{2}}$}, and $\braket{\beta}{10} = \braket{\beta Z}{10} = 0$.
Therefore projection of $\ket{10}$ to the Bell basis gives
 $\ket{X \beta}$ and $\ket{XZ \beta}$ with equal probability, corresponding to the measurement
results $(1,0)$ and $(1,1)$ reported by our gadget, but
now \emph{these} Bell basis states are the two possible results of measurement.
One can extend the gadget so that it truly `measures in the Bell basis' by having it apply
the reverse of the measurement circuit, as shown below,
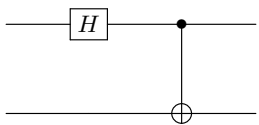
\[ \Qcircuit{ & \gate{H} & \ctrl{1} & \qw  \\ & \qw & \targ{} & \qw  }\]
 to each of the four
possible states $\ket{0}\!\ket{0}$, $\ket{0}\!\ket{1}$, $\ket{1}\!\ket{0}$, $\ket{1}\!\ket{1}$ that 
are the outputs of the original gadget.}

\smallskip
\begin{exercise}
Suppose that Alice and Bob share a Bell pair. 
Show that Alice can transmit two classical bits of information to Bob by applying
suitable Pauli operators to the qubit in her possession and then sending this single qubit to Bob.\footnote{\textbf{Solution:} Alice applies either $I$, $X$, $Z$ or $XZ$. Bob receives the qubit
and then measures both his qubits using the gadget above. By the table, his pair of measurements
distinguishes the four cases.} 
\end{exercise}

Superdense encoding is perhaps only surprising if you think that a qubit is an ordinary classical
bit in disguise and overlook that entanglement is a vital resource in quantum communication.
A much more striking application of Bell state preparation and measurement is
quantum teleportation, but this is beyond the scope of this section 
because it needs three qubits.\footnote{\textbf{But anyway:} 
to demonstrate quantum teleportation algebraically,
use the identity
$\bigl(\alpha\!\ket{0}+\beta\!\ket{1}\bigr) \mfrac{1}{\sqrt{2}} \bigl( \ket{00} + \ket{11} \bigr)
= \frac{1}{2}\sum_{P \in \{I,X,Z,XZ\}} %note 1/2 not 1/4
\mfrac{1}{\sqrt{2}} P\bigl( \ket{00} + \ket{11} \bigr) P \bigl(\alpha\!\ket{0}+\beta\!\ket{1}\bigr)$
to show that when Alice measures the first two qubits 
of
\[ \bigl(\alpha \!\ket{0} + \beta\!\ket{1}\bigr)\ket{\Phi} = \mfrac{1}{\sqrt{2}} \bigl(\alpha\! \ket{0} + \beta\!\ket{1}\bigr)\bigl(
\ket{00} + \ket{11} \bigr) \]
by projecting to the Bell basis (see footnote~\ref{footnote:pitfall} `Pitfall' above), the result is 
\[ P\ket{\Phi}=  \mfrac{1}{\sqrt{2}} P \bigl( \ket{00} + \ket{11}  \bigr) P\bigl(
\alpha \! \ket{0} + \beta\!\ket{1}\bigr),\]
where $P$ is one of $I$, $X$, $Z$ or $XZ$ each with equal probability. Thus Bob now has the qubit $\alpha \!\ket{0} + \beta\!\ket{1}$
originally held by Alice, up to a Pauli correction that Alice will have to send him by two classical bits;
this correction is the only reason why quantum teleportation does not violate causality.
(Look up `No communication theorem' for a precise formulation of this,
and see also \S\ref{subsec:noCloning} for the related `No cloning theorem'.)
Quantum teleportation is more elegant in the ZX-calculus (you can easily find this on the web,
but of course~I like 
the account in my blog post \url{https://wildonblog.wordpress.com/2021/10/24/q-is-for-quantum/} reviewing
Rudolph's book \emph{$Q$ is for quantum} \cite{RudolphQ})
and exceptionally elegant in the string-diagram formalism, where it becomes
the basic yank relation: see for instance \cite[\S 3c]{Coecke05}.}
% Omitted details of the calculation, taken from my blog post

\subsection{The failure of spin polarisation${}^\star$}
Spin polarisation is the principle that any single qubit\footnote{\textbf{Lie:} 
this is only true for particles of spin $\mfrac{1}{2}$ such as the electron. Photons are mathematically
an equally good model for a qubit, but physically behave differently, 
because of circular polarisation.} %diagonal polarisation X is still a superposition of horizontal and vertical
has a well-defined spin in some direction $\vecn{n}$ in $\R^3$. \marginpar{\raisebox{-18pt}{\scalebox{0.9}{$\begin{matrix}
Z = \left(\begin{matrix} 1 & 0 \\ 0 & -1 \end{matrix}\right) \\[12pt]
X = \left(\begin{matrix} 0 & 1 \\ 1 & 0 \end{matrix}\right) \\[12pt]
Y \!=\! iXZ \!=\! \left(\begin{matrix} 0 & -i \\ i & 0 \end{matrix}\right) \end{matrix}$}}}
We proved it in Exercise~\ref{ex:spin} by showing that 
the normalized state $\alpha\! \ket{0} + \beta\! \ket{1}$
is an eigenstate of $n_Z Z + n_X X + n_Y Y$ for suitable $(n_Z, n_X, n_Y) \in \R^3$.
That this exercise is non-trivial should perhaps warn us that this principle is not nearly as obvious
as it might seem from our naive idea of magnetism or spinning tops. (Hence all the fuss about
the Bloch sphere in \S\ref{subsec:Bloch}.)
The Bell state shows that spin polarisation already fails for two qubits, or more.

\begin{exercise}
Alice has the first qubit of a Bell state and Bob the second.
Show that the expected value of an Alice measurement of $Z$ on the first qubit is zero.
(The general definition of $Z$-basis measurement in Definition~\ref{defn:measureZpartial}, but
you will probably guess the right thing to do: it is mathematically intuitive, even if the physical
consequences are less so!)
Show that the same holds if Alice switches from measurement in the $\vecn{z}$ direction to 
an arbitrary direction $\vecn{n} \in \R^3$.\footnote{\textbf{Solution:} 
when Alice measures the first qubit of \smash{$\mfrac{1}{\sqrt{2}} \bigl( \ket{00} + \ket{11}\bigr)$} 
in the $Z$-basis, she projects to $\ket{00}$ and $\ket{11}$ with equal probability;
the measurement results are $0$ and $1$, corresponding to eigenvalues $+1$ and $-1$ respectively.
Therefore her results are flat random and the expectation is $0$.
% measuring by Hermitian operator K gives <\psi | K | \psi> projecting to K\psi, e.g. for Z
% <00 + 11 | 00 - 11> = 0
More generally, using the formula $\braket{\psi}{K\psi}$ for the expectation value
of the Hermitian operator $K$ on the normalized state $\ket{\psi}$ we have
\[ \begin{split} \langle \psi | n_X (X \otimes I) +{}&{} n_Y (Y \otimes I) + n_Z (Z \otimes I) \psi \rangle
\\ &= n_X \langle \psi | X \psi\rangle + in_Y \langle \psi | XZ \psi \rangle + n_Z \langle \psi | Z \psi\rangle
= 0 + 0 + 0 = 0; \end{split} \] 
this symmetry of the Bell state becomes less surprising if one calculates
that \smash{$\ket{\Phi} = \mfrac{1}{\sqrt{2}} \bigl(\ket{+}\ket{+} + \ket{-}\ket{-}\bigr)$}; on this
calculation, I cannot resist
quoting Maudlin's book \emph{Philosophy of physics: Quantum theory} 
\cite[page 71, footnote 6]{Maudlin}: `\emph{It's just a little painless algebra that anyone can do, and
doing it produces a sense of accomplishment and understanding that can be acquired in no other way}'.}
%Incidentally the bra-ket pair above is often separated as $\langle \psi|K|\psi\rangle$ to emphasise
%the symmetry: but note however that $\langle \psi | U $ is the dual element corresponding to $U^\dagger \ket{\psi}$; this is not a problem in the physicists bra-ket notation because invariably when one 
%sees $\langle \psi|K|\psi \rangle$ it is in the context of a Hermitian $K$ and so $K = K^\dagger$.
Resolve the apparent paradox, as stated in \cite[page 167]{SusskindFriedmanQM}: `How can that be? How could we know
\emph{as much as can possibly be known} about the Alice-Bob system of two spins,
and yet know \emph{nothing} about the individual spins that are its 
subcomponents?'\footnote{\textbf{Solution(?):} for this author, a satisfactory resolution of the paradox is to declare
that, in the context of the Bell state --- and more generally in any state not factorizing as a tensor product of a state in $\HH$ and a state in $\HH^{\otimes (n-1)}$ --- 
there is not such thing as an individual spin. This is not satisfactory to many physicists,
since they are well used to measuring exactly this quantity;
see footnote~\ref{footnote:statesMeasurements} for a less precise
variation on this theme.} I took `the theoretical minimum'
in the title for these notes from the title of
\cite{SusskindFriedmanQM}, and can highly recommend it
as an introduction to quantum mechanics.
\end{exercise}

\begin{exercise}\label{ex:BellObservables}
Show that $E = X \otimes X + Y \otimes Y + Z\otimes Z$ is an observable of 
states in the Bell basis (or stated mathematically, the Bell states are its eigenvectors)
and that its eigenvalues distinguish the \emph{spin singlet}
state \smash{$\mfrac{1}{\sqrt{2}} \bigl( \ket{01} - \ket{10}\bigr)$} from the other three elements
of the Bell basis.
What is the representation theoretic interpretation of this?\footnote{\textbf{Answer:}
we have seen that $X \otimes X$ and $Z \otimes Z$ are stabilisers of $\ket{\Phi}$. Hence
so is $(X \otimes X)(Z \otimes Z) =  -iXZ \otimes iXZ =-Y \otimes Y$
and it follows that $E \ket{\Phi} = (1+1-1) \ket{\Phi} = \ket{\Phi}$. Using that distinct Pauli
matrices anticommute, we have 
\begin{align*} E(X \otimes I) 
&= (X \otimes X)(X\otimes I) + (Y \otimes Y)(X \otimes I) + (Z \otimes Z)(X\otimes I) \\
&= (X \otimes I) (X \otimes X) - (X \otimes I) (Y \otimes Y) - (X \otimes I) (Z \otimes Z),\end{align*}
and so, using the unambiguous notation of the previous subsection,
$E(X \ket{\Phi}) = X \ket{\Phi} + X \ket{\Phi} - X \ket{\Phi} = X\ket{\Phi}$.
Similarly $E(Z \ket{\Phi}) = Z\ket{\Phi}$. 
For $XZ$ we have
\[ E (XZ \otimes I) = - (XZ \otimes I)(X \otimes X) + (XZ \otimes I) (Y\otimes Y) - (XZ \otimes I)(Z \otimes Z) \]
and so $E(XZ \ket{\Phi}) = -XZ\ket{\Phi} - XZ\ket{\Phi} - XZ\ket{\Phi} = -3XZ\ket{\Phi}$.
Therefore $E$ has eigenvalue $-3$ on the spin singlet state $XZ\ket{\Phi}$ 
and eigenvalue $1$ on the other Bell basis states.
The representation-theoretic significance is that $E$ commutes with the action
of $\SU_2$ on $\HH \otimes \HH$ and so, as expected from Schur's Lemma,
its eigenspaces split $\HH \otimes \HH$ 
into its two irreducible subrepresentations, namely $\bigwedge^2\hskip-1pt \HH = \bigl\langle \hskip1pt \ket{01} - \ket{10} \bigr\rangle$
and $\Sym_2 \!\HH = \bigl\langle \hskip1pt \ket{00}, \ket{11}, \ket{01} + \ket{10} \bigr\rangle$.
If you have ever sat through a lecture on $\mathsf{sl}_2(\C)$ in which the lecturer pulled the Casimir
element $\mfrac{1}{2} h \otimes h + e \otimes f+f \otimes e \in \mathcal{U}(\mathsf{sl}_2(\C))$ out of a hat, 
you might be amused that
interpreted in $\su_2$, it is, up to a scalar, %$(iX)\otimes(iX) + (iY)\otimes (iY) + (iZ) \otimes (iZ) = 
$X \otimes X + Y \otimes Y + Z \otimes Z$, lying in the quadratic component of the universal enveloping algebra
$\mathcal{U}(\su_2)$. Here the $\mathsf{sl}_2(\C)$-generators are
\[ h = \left(\begin{matrix} 1 & 0 \\ 0 & -1 \end{matrix} \right), \ e = \left(\begin{matrix} 0 & 1 \\ 0 & 0 \end{matrix}\right)
,\ f = \left(\begin{matrix} 0 & 0 \\ 1 & 0 \end{matrix}\right), \]
and the claim follows from the isomorphism $\su_2 \otimes_\mathbb{R} \C \cong \mathsf{sl}_2(\C)$
defined by $-\mfrac{iZ}{2} \mapsto \mfrac{ih}{2}$, $-\mfrac{iX}{2} \mapsto \mfrac{1}{2}(e-f)$, 
\smash{$-\mfrac{iY}{2} \mapsto \mfrac{i}{2}
(e+f)$}, where the $\su_2$ generators are as in footnote~\ref{footnote:adjointRep}.
This is one of the best examples the author knows where a physical motivation has a clear benefit
in a non-trivial pure mathematics problem.} % possible extra Casimir in sl_2 or su_2 ...
\end{exercise}

\section{Many qubits: quantum computation}\label{sec:manyQubits}

In this section we show that quantum computers are amazingly good at computing the Discrete Fourier
Transform. From the point of view of this section, which emphasises
the symmetry between the $Z$-basis and $X$-basis,
this comes down to a simple change-of-basis.
As a corollary we prove that a quantum computer can solve
certain problems on Boolean functions
using exponentially fewer logic gates than classical computers.

\subsection{$Z$-basis} \label{subsec:Zbasis}
The correct
 way to model $n$ qubits is by the tensor product $\HH \otimes \cdots \otimes \HH$ which we
usually denote $\HH^{\otimes n}$. 
Given $v \in \mathbb{F}_2^n$ let 
\[ \ket{v} = \ket{v_1} \otimes \ket{v_2} \otimes \cdots \otimes \ket{v_n} \in \HH^{\otimes n}. \]
Thus $\HH^{\otimes n}$ is a $2^n$-dimensional space having as an orthonormal basis $\ket{v}$ for $v \in \mathbb{F}_2^n$.
This is  the $Z$-\emph{basis} of $\HH^{\otimes n}$ of eigenvectors for $Z \otimes Z \otimes \cdots \otimes Z$.
Note this notation is consistent with our earlier $\ket{0}, \ket{1} \in \HH$ and $\ket{00}, \ket{01}, \ket{10},
\ket{11} \in \HH \otimes \HH$. We emphasise the distinction between the two vector spaces $\HH^{\otimes n}$ and $\F_2$:
\begin{exlist}
\item[$\bullet$] $\HH^{\otimes n}$ is the $2^n$-dimensional complex Hilbert space modelling $n$  qubits;
\item[$\bullet$] $\F_2^n$ is a finite $\F_2$-vector space whose elements index the $Z$-basis of~$\HH^{\otimes n}$.\footnote{\textbf{Aside, please ignore if it is not helpful to you:}
a slightly similar indexing scheme occurs in the theory of $\theta$ functions where
lattice points in $\R^2$ index an orthonormal basis of an infinite dimensional Hilbert space.}
\end{exlist}
We say that a state $\ket{\psi}$ is \emph{normalized} if $\braket{\psi}{\psi} = 1$, or equivalently,
if \smash{$\sum_{v \in \F_2^n} |a_v|^2 = 1$}, where \smash{$\ket{\psi} = \sum_{v \in \F_2^n} a_v \ket{v}$}.

\begin{definition}[$Z$-basis measurement of all qubits]\label{defn:measureZall}
Let \smash{$\ket{\psi} = \sum_{v \in \F_2^n} a_v \ket{v}$} be a normalized state.  \emph{Measuring $\ket{\psi}$
in the $Z$-basis} projects it to $\ket{v}$ with probability $|a_v|^2$. The measurement result is $v \in \F_2^n$.
\end{definition}

Again let us observe the traditional pause to remember that \emph{measurement changes the state}.

\subsection{Transverse Hadamard}\label{subsec:transverseHadamard} 
We saw earlier that the one-qubit Hadamard gate satisfies $H\ket{0} = \ket{+}$ and $H\ket{1}= \ket{-}$.
It will now be useful to combine the cases by writing 
\[ H \!\ket{b} = \mfrac{1}{\sqrt{2}} \bigl( \ket{0} + (-1)^b
\ket{1}\bigr). \tag{$\star$} \]
for $b \in \F_2$.
The $n$-fold tensor product $H \otimes \cdots \otimes H = H^{\otimes n}$ is a unitary transformation
of $H^{\otimes n}$; it is sometimes called \emph{transverse Hadamard} to emphasise that one $H$ gate is applied
to each qubit separately. Let $\cdot$ denote the dot product on~$\F_2^n$, defined as usual by $v \cdot w 
= v_1w_1 + \cdots + v_n w_n \in \F_2$.

\begin{lemma}[Transverse Hadamard]\label{lemma:transverseHadamard}
Given $v \in \F_2^n$ we have
\[ H^{\otimes n} \ket{v} = \frac{1}{2^{n/2}} \sum_{w \in \F_2^n} (-1)^{v \cdot w} \ket{w}. \]
\end{lemma}
\begin{proof}
Using ($\star$) for the second equality we have
\begin{align*}
H^{\otimes n} \ket{v} &= H \ket{v_1} \otimes \cdots \otimes H \ket{v_n} \\
&= \mfrac{1}{\sqrt{2}} \bigl( \ket{0} + (-1)^{v_1} \ket{1} \bigr) \otimes \cdots \mfrac{1}{\sqrt{2}} 
\bigl( \ket{0} + (-1)^{v_n} \ket{1} \bigr) \\
&= \frac{1}{2^{n/2}} \Bigl( \sum_{w_1 \in \F_2}  (-1)^{v_1w_1} \ket{w_1} \Bigr) 
\otimes \cdots \otimes \Bigl( \sum_{w_n \in \F_2}  (-1)^{v_nw_n} \ket{w_n} \Bigr) \\
&= \frac{1}{2^{n/2}}\sum_{w_1, \ldots, w_n \in \F_2} (-1)^{v_1w_1 + \cdots + v_nw_n} \ket{w_1} \ldots \ket{w_n} \\
&= \frac{1}{2^{n/2}}\sum_{w \in \F_2^n} (-1)^{v \cdot w} \ket{w} 
\end{align*}
as required.
\end{proof}

\subsection{$X$-basis}\label{subsec:Xbasis}
The $X$-\emph{basis} of $\HH^{\otimes n}$ is $H^{\otimes n} \ket{v}$ for $v \in \F_2^n$. % Do we want the ${}^\pm$ notation?
For example the $X$-basis of $\HH^{\otimes 2}$ is shown left below, with the change of basis matrix $H^{\otimes 2}$
to the right, using our usual order $\ket{00}, \ket{01}, \ket{10}, \ket{11}$.
\[ \begin{matrix}
\ket{+}\ket{+} = H^{\otimes 2} \ket{00} = \mfrac{1}{2} \bigl( \ket{00} + \ket{01} + \ket{10} + \ket{11} \bigr) \\[3pt]
\ket{+}\ket{-} = H^{\otimes 2} \ket{01} = \mfrac{1}{2} \bigl( \ket{00} - \ket{01} + \ket{10} - \ket{11} \bigr) \\[3pt]
\ket{-}\ket{+} = H^{\otimes 2} \ket{10} = \mfrac{1}{2} \bigl( \ket{00} + \ket{01} - \ket{10} - \ket{11} \bigr) \\[3pt]
\ket{-}\ket{-} = H^{\otimes 2} \ket{11} = \mfrac{1}{2} \bigl( \ket{00} - \ket{01} - \ket{10} + \ket{11} \bigr) \\

\end{matrix}
\quad\mfrac{1}{2}\left( \begin{matrix} 1 & 1 & 1 & 1 \\ 1 & -1 & 1 & -1 \\ 1 & 1 & -1 & -1 \\ 1 & -1 & -1 & 1 \end{matrix}
\right)
\]
Up to the scalar $2^{-n/2}$, $H^{\otimes n}$ is the character table of $\F_2^n$, having
as its rows the characters of $\F_2^n$. Depending
on your background, this may make the application to Boolean functions below  seem almost inevitable;
if not, well great, you can now regard a large part of classical cryptanalysis as an application of 
basic character theory.

\subsection{Controlled gates}
Given a Boolean function $f : \F_2^n \rightarrow \F_2$ 
on $n$ bits,  how can we implement $f$ as a quantum gate? 
For example, what unitary transformation will encode logical AND, defined by $f(x,y) = xy$?
It is well worth finding the answer for yourself. As a hint, consider the case $n=1$ aiming
to implement NOT using two qubits. Remember that since quantum gates are unitary maps, and so
invertible, the input bits must be recoverable from the output.

\begin{definition}[Controlled gates]\label{defn:controlledGate}
Given \smash{$f : \F_2^n \rightarrow \F_2$} a Boolean function, \emph{controlled $f$} is the unitary map on $n+1$
qubits defined on the $Z$-basis of $\HH^{\otimes (n+1)}$ by
\[ U_f \ket{v} \ket{b} = \ket{v} \ket{b + f(v)} \]
where $v \in \F_2^n$ and $b\in \F_2$.
\end{definition}

Observe that since the $Z$-basis of $\HH^{\otimes (n+1)}$ is permuted by $U_f$, this map is unitary, as required.
Taking $n=1$ and $f(x) = x$ we recover the CNOT
operator\marginpar{
\raisebox{114pt}{\Qcircuit @C = 24pt @R = 15pt{ 
  &  \ctrl{3} & \qw \\ 
  & \ctrl{2} & \qw \\
  & \ctrl{1} & \qw \\  
 & \gate{U_f} & \qw }}} 
 and taking $n=2$ and $f(x,y) = xy$ gives the Toffoli
 gate (see \cite[page 159]{NielsenChuang}) that is the quantum
 analogue of logical AND.
 The circuit diagram in the margin shows how to draw $U_f$ when $n=3$;
 the three black dots indicate that 
the output qubit is controlled on all $3$ input qubits. 
For general $n$ a tick is used
to indicate there are $n$ input wires.
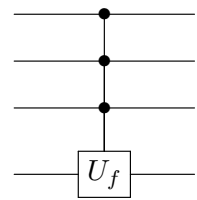
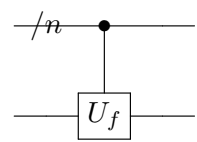
\marginpar{\raisebox{16pt}{\Qcircuit @C=12pt { & \qw{/n} & \ctrl{1} & \qw & \qw\\ 
& \qw & \gate{U_f} & \qw & \qw }}}

If you would like to find your own path to Theorem~\ref{thm:QDFT} a good start is
the following question. For the author, at least half its interest lies in the fact that it can be stated
\emph{at all}: do programs in your favourite classical programming language
have eigenvectors?

\begin{exercise}
What are the eigenvectors 
and eigenvalues of $U_f$?\footnote{\textbf{Solution:} 
for each $w \in \F_2^n$, $\ket{w}\ket{+}$ is an eigenvector with eigenvalue $1$;
$\ket{w}\ket{-}$ is an eigenvector with eigenvalue \smash{$(-1)^{f(w)}$}. This is part of the
calculation in \S\ref{subsec:phase}.}
\end{exercise}

\subsection{Graph states}
Generalizing Exercise~\ref{ex:makeBell}, considering what happens when we apply $U_f$ to an input
that is not a $Z$-basis state, but instead a superposition, such as $\ket{+}\ket{+}$ seen in \S\ref{subsec:Xbasis}.
As in this exercise, it is immediate from linearity that
\[ U_f \bigl( \ket{+} \ldots \ket{+} \ket{0} \bigr) = \mfrac{1}{2^{n/2}} \sum_{v \in \F_2^n} U_f \ket{v} \ket{0} =
\mfrac{1}{2^{n/2}} \sum_{v \in \F_2^n} \ket{v}\ket{f(v)}. \]
Thus it is tempting to say that a quantum computer computes `all the values of a function at once'. The state above
is the quantum encoding of the graph of $f$.
For the circuit diagram see Figure~\ref{fig:Uf} below.

\begin{exercise}
What is the result of measuring the graph state above in the $Z$-basis? Are you yet persuaded
to splash out on a quantum computer?\footnote{\textbf{Solution:}
by Definition~\ref{defn:measureZall}, the measurement projects the state to
some $\ket{v}\ket{f(v)}$ where $v$ is distributed
uniformly at random on $\F_2^n$. This is clearly \emph{worse} than a classical computer, because you don't
even get to choose~$v$.}
\end{exercise}

Now see what happens if we change to the $X$-basis \emph{everywhere}.

\subsection{Moving the output to the phase}\label{subsec:phase}
A standard trick in quantum computing is to move a computational result such as $\ket{f(v)}$ to the phase
of its qubit, by replacing an input $\ket{0}$ with $\ket{-}$: thus
\[ \begin{split} U_f \ket{v}\!\ket{-} &=
\mfrac{1}{\sqrt{2}} U_f \bigl( \ket{v}\!\ket{0} \hskip1pt-\hskip1pt \ket{v}\!\ket{1}\bigr) \\
&\qquad = \mfrac{1}{\sqrt{2}} 
\bigl( \ket{v}\!\ket{f(v)}
\hskip1pt-\hskip1pt \ket{v}\bigl| \hskip1pt \overline{f(v)}\hskip1pt \bigr\rangle \bigr) = (-1)^{f(v)} \ket{v}\! \ket{-} 
\end{split} \]
and so, as shown in the middle circuit in Figure~\ref{fig:Uf},
\[ U_f \bigl( \ket{+} \ldots \ket{+} \ket{-} \bigr) =
\mfrac{1}{2^{n/2}} \sum_{v \in \F_2^n} U_f \ket{v} \ket{-} =
\mfrac{1}{2^{n/2}} \sum_{v \in \F_2^n} (-1)^{f(v)} \ket{v}\ket{-}. \]

\subsection{The Quantum Discrete Fourier Transform in $\F_2^n$.}\label{subsec:QDFT}
We now move the output to the $X$-basis as well by applying a final transverse Hadamard on the first $n$-qubits.
Calculating using Lemma~\ref{lemma:transverseHadamard} and the previous result gives
\begin{align*} (H^{\otimes n} \otimes I) U_f \bigl( \ket{+} \ldots \ket{+} \ket{-} \bigr)
&= \frac{1}{2^{n/2}} \sum_{v \in \F_2^n} (-1)^{f(v)} \bigl( H^{\otimes n}\ket{v}\bigr) \ket{-} \\
&= \frac{1}{2^{n/2}} \sum_{v \in \F_2^n} (-1)^{f(v)} \sum_{w \in \F_2^n} (-1)^{v \cdot w}\ket{w} \ket{-} \\
&= \sum_{w \in \F_2^n} \frac{1}{2^n} \Bigl( \sum_{v \in \F_2^n} (-1)^{f(v) + v\cdot w}\Bigr) \ket{w} \ket{-} \tag{$\dagger$}.
\end{align*}
The quantity $\frac{1}{2^n}\sum_{v \in \F_2^n} (-1)^{f(v) + v \cdot w}$ is the \emph{correlation} between
$f$ and the linear map $v \mapsto v \cdot w$. We shall denote it $c_f(w)$. Observe that
\[ c_f(w) = \PP_v[f(v) = v \cdot w] - \PP_v[f(v) \not= v \cdot w] \] 
where the probability is with respect to $v$ uniformly distributed over $\F_2^n$. 
Since $H^{\otimes n} \otimes I$ is unitary, and so norm preserving,
we have \smash{$\sum_{v \in \F_2^n} c_f(v)^2 = 1$}.
Correlations even slightly greater than the average $2^{-n/2}$ can often be used
as the basis of a cryptographic attack, but, using only a classical computer,
 it appears there is no efficient way to find them.

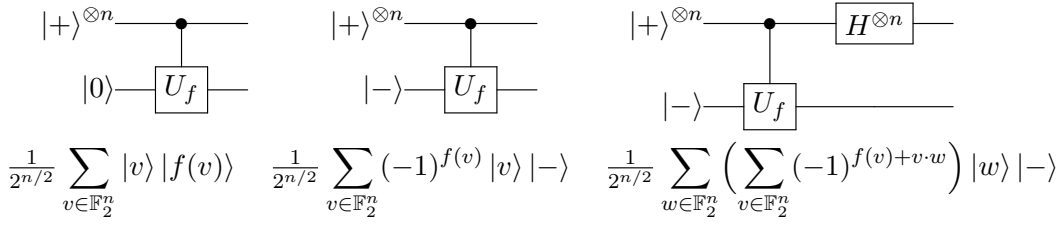
\begin{figure}[t]
\begin{align*}& \hspace*{-5pt}\qquad \Qcircuit @R=15pt @C=15pt{ & \llap{$\ket{+}^{\otimes n}$} & \ctrl{1} & \qw \\
  & \llap{$\ket{0}$}  & \gate{U_f} & \qw }  
 \qquad\qquad 
  \Qcircuit @R=15pt @C=15pt{ & \llap{$\ket{+}^{\otimes n}$} & \ctrl{1} & \qw \\
  & \llap{$\ket{-}$}  & \gate{U_f} & \qw }   
  \qquad\qquad \
    \Qcircuit @R=15pt @C=15pt{ & \llap{$\ket{+}^{\otimes n}$} & \ctrl{1} & \gate{H^{\otimes n}} & 
    \qw \\
  & \llap{$\ket{-}$}    & \gate{U_f} & \qw  & \qw }  \\  
  &  \hspace*{-9pt} \mfrac{1}{2^{n/2}} \sum_{v \in \F_2^n} \ket{v}\ket{f(v)}\ \ \
  \mfrac{1}{2^{n/2}} \sum_{v \in \F_2^n} (-1)^{f(v)} \ket{v}\ket{-}  
    \ \ \ \mfrac{1}{2^{n/2}} \sum_{w \in \F_2^n} \Bigl( \sum_{v \in \F_2^n} (-1)^{f(v) + v\cdot w} \Bigr)
  \ket{w}\ket{-} 
  \end{align*}
\caption{\small Final outputs of circuit for $U_f$ where $f : \F_2^n \rightarrow \F_2$
is a Boolean function: (1) graph circuit; (2) result moved to phase; (3) with $H^{\otimes n}$,
computing $\sum_{w \in \F_2^n} c_f(w) \ket{w} \ket{-}$. \label{fig:Uf}}
\end{figure}

\begin{theorem}[Quantum Discrete Fourier Transform in $\F_2^n$]\label{thm:QDFT}
Let $f : \F_2^n \rightarrow \F_2$ be a Boolean function.
There is a quantum circuit with input \smash{$\ket{0}^{\otimes (n+1)}$}, using 
one $X$ gate, $2(n+1)$ Hadamard gates and one controlled $U_f$ gate,
that prepares the state \smash{$\sum_{v \in \F_2^n} c_f(v) \ket{v}\ket{1}$}.
Measuring this state in the $Z$-basis returns $\ket{v}\ket{1}$ 
with probability $c_f(v)^2$.
Computing the exact value of any single $c_f(v)$ by a classical computer requires $2^n$ evaluations of $f$.
\end{theorem}

\begin{proof}
Since 
\[ H^{\otimes (n+1)} (I^{\otimes n} \otimes X) \ket{0}^{\otimes n} = (H^{\otimes n} \otimes H)
\bigl(\ket{0}^{\otimes n} \otimes \ket{1}\bigr) = \ket{+}^{\otimes n} \ket{-} \]
we can prepare 
the input state for the rightmost circuit in Figure~\ref{fig:Uf} using~$n$ Hadamard gates and one $X$-gate.
The circuit itself uses one $U_f$ gate and~$n$ Hadamard gates, and we use one further Hadamard gate
to switch the final qubit $\ket{-}$ to $\ket{1}$, just to make the final $Z$-basis measurement deterministic
on the final qubit.
The theorem now follows from the definition of correlation and the previous calculation~($\dagger$).
\end{proof}

%Using the Fast Fourier Transform, one can compute all $c_f(v)$ for $v \in \F_2^n$ on a classical
%computer in order of $2^n n$ operations. 
The circuit in Theorem~\ref{thm:QDFT} is 
shown left below. On the right is the equivalent version where we
prepare and measure in the $X$-basis. Its simplicity is a compelling demonstration of the power
of changing the basis!

\smallskip
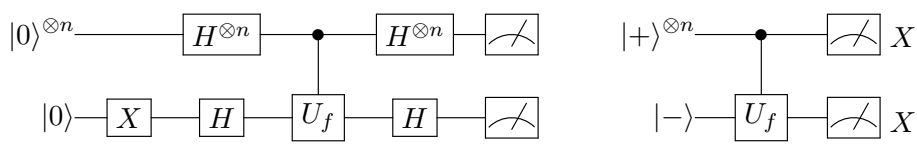
\[ \hspace*{0.1in}
\Qcircuit @R=15pt @C=12pt{ & \llap{$\ket{0}^{\otimes n}$} & \qw & \gate{H^{\otimes n}} & \ctrl{1} & \gate{H^{\otimes n}}  
    & \meter{} \\
  & \llap{$\ket{0}$} & \gate{X} & \gate{H}   & \gate{U_f} & \gate{H}  &  \meter{} }  
  \qquad\qquad
\Qcircuit @R=15pt @C=15pt{ & \llap{$\ket{+}^{\otimes n}$} & \ctrl{1} & \meter \\
  & \llap{$\ket{-}$} & \gate{U_f}   &  \meter{} } \ \raisebox{-18pt}{$\begin{matrix} X \\[18pt] X \end{matrix}$}
\]

\smallskip
\subsection{Deutsch--Jozsa}
You are given $f : \F_2^n \rightarrow \F_2$ and told that \emph{either}
\begin{exlistB}
\item[$\bullet$] $f$ is constant \emph{or}
\item[$\bullet$] $f$ is balanced, equal to $0$ and $1$ with equal probability.
\end{exlistB}
The output of the circuit for Theorem~\ref{thm:QDFT} is $\ket{0} \ket{0}$  
in the first case, and, since a balanced function has zero correlation
with the constant `all-ones' function, $\ket{w}\ket{0}$ for some $w \not=0$ in
the second case. (Note this is after measuring.) Thus only one use of $U_f$ is required
to distinguish the two cases.
To decide classically requires, in the adversarial worst case, $2^{n-1} + 1$
evaluations of $f$. While of little (no?)~practical 
importance, the Deutsch--Jozsa Problem was
one of the first examples of a problem
that could be solved by a quantum algorithm \emph{exponentially faster}
than any deterministic classical algorithm.\footnote{\textbf{Objection:} 
three evaluations of $f$ give a classical algorithm guaranteed to succeed
when $f$ is constant, and with error probability 
$\mfrac{1}{4}$ when $f$ is balanced. This error probability can, as usual,
be made arbitrarily small by repeated runs of the algorithm. Thus the Deutsch--Jozsa problem
is in \textsf{BPP}. Simon's Problem shows that 
a probabilistic quantum algorithm (in the class \textsf{BQP}) 
may be exponentially faster than any probabilistic classical algorithm; 
again this uses the Quantum Discrete Fourier Transform
on $\F_2^n$ but now with~$n$ output bits: see \cite[\S 2.5]{Mermin}.}

\subsection{Shor's algorithm}
Shor's algorithm makes an ingenious use of the Quantum Discrete Fourier Transform,
but now working in the group $\Z/2^n\Z$, to find the period of, for instance, the doubling
function $x \mapsto 2^x$ mod $N$, and so (by an easy classical endgame that
uses a largish prime factor of $\phi(N)$) factor $N$.
Because the Fourier transform has to be taken in $\Z/2^n\Z$ rather than $\Z/\phi(N)\Z$,
there is a small probability of error\footnote{\textbf{Why:} for RSA numbers $N=pq$, 
because finding $\phi(pq) = pq-p-q+1$
is exactly as hard as factoring $pq$}; this can be made exponentially small
by running the algorithm polynomially many times.\footnote{\textbf{Reference:} 
see \cite[\S 5.1]{NielsenChuang}. Or better, read the easier and more motivated
following section \S 5.2 on phase estimation, and then work out the circuit in \S 5.1 for yourself. 
For an excellent exposition of the ideas
that follows this order, see \S 10.8 and \S 10.9 in \cite{EkertHosgoodKayMacchiavello}.}

\subsection{\textsf{BQP} and complexity theory${}^\star$}\label{subsec:BQP}
The complexity class \textsf{EXACTQP} is, roughly put, all problems solvable on a quantum computer 
starting at the $\ket{0}^{\otimes n}$ state for some chosen $n$, applying polynomially
many gates --- drawn from a fixed universal gate set\footnote{\textbf{Universal gate set:} \label{footnote:SolovayKitaev}
a popular choice is the $S$ gate (which satisfies $S^2=Z$), the Hadamard gate $H$, the
CNOT gate, and one extra gate, such as the $T$-gate
(which satisfies $T^2 = S$), or the $3$-qubit Toffoli gate; this extra gate makes the step from the finite Clifford subgroup
of $U_{2^n}$ to a countable dense subgroup of $U_{2^n}$. Roughly put, the Solovay--Kitaev theorem states
that an arbitrary unitary gate can be well-approximated by an efficient \emph{and efficiently findable}
composition of gates from the chosen generating set. Still
this approximation process introduces some overhead; 
this is one motivation for
measurement-based quantum computation, in which the overhead 
reappears as a sometimes more manageable cost
in making ancilla states.}
 --- and then
measuring the result in the $Z$-basis. (The principle of deferred measurement implies that,
perhaps contrary to one's intuition, allowing multiple measurements does not boost the computational power.\footnote{\textbf{Caveat:} provided no errors occur. Repeated measurement is vital to quantum error
correction, see \S\ref{subsec:measuringStabiliser}.})
Because quantum computers can be used to simulate
classical computers, \textsf{EXACTQP} contains \textsf{P}. Shor's algorithm lies in the bigger class
\textsf{BQP} in which a bounded probability of error is allowed. Thus Shor's algorithm
shows that \textsf{BQP} contains a problem widely believed to be in $\textsf{NP}$ but not in $\textsf{P}$.
It is possible that $\textsf{BQP}$ contains $\textsf{NP}$, or that $\textsf{NP}$ contains $\textsf{BQP}$,
and both possibilities are consistent with 
$\textsf{P} = \textsf{NP}$. It is much more widely believed
that neither of $\textsf{BQP}$ nor $\textsf{NP}$ contains the other; in this
case $\textsf{P} \not= \textsf{NP}$ and, as shown in Figure~\ref{fig:comp}, factoring is an example of a problem in
the intersection of $\textsf{BQP}$ and $\textsf{NP}$ that may not also be in $\textsf{P}$.\footnote{\textbf{Reference:} for much more on this
see \cite{Aaronson}.}

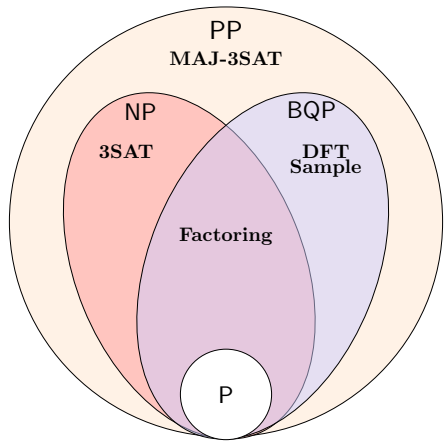
\begin{figure}
\begin{center}
\scalebox{0.75}{\begin{tikzpicture}

    \draw[fill=orange!10] (0.0, 1.25) circle (3.8cm);
    \node[scale=1.2] at (0, 4.65) {\textsf{PP}};

    \begin{scope}[rotate=22]
    \draw[fill=red!40, fill opacity=0.5] (-0.42, 0.68) ellipse (2.0cm and 3.2cm);
    \end{scope}
    \node[scale=1.1] at (-1.5, 3.2) {\textsf{NP}};

    \begin{scope}[rotate=-22]
    \draw[fill=blue!20, fill opacity=0.5] (0.42, 0.68) ellipse (2.0cm and 3.2cm);
    \end{scope}
    \node[scale=1.1] at (1.5, 3.2) {\textsf{BQP}};

    \draw[fill=white] (0,-1.8) circle (0.8cm);
    \node[scale=1.1] at (0, -1.8) {\textsf{P}};

    \node[scale=0.9, align=center] at (0, 1.0) {\textbf{Factoring}};
    \node[scale=0.9, align=center] at (1.75, 2.5) {\textbf{DFT}};
    \node[scale=0.9, align=center] at (1.75, 2.2) {\textbf{Sample}};
    \node[scale=0.9, align=center] at (-1.75, 2.5) {\textbf{3SAT}};
    \node[scale=0.9, align=center] at (0, 4.15) {\textbf{MAJ-3SAT}};

\end{tikzpicture}}\end{center}
\caption{\small One possible configuration of the complexity theory landscape. \textsf{PP} is (although not
by the usual definition) the class
of decision problems solvable in polynomial time with error probability
for \emph{either} answer strictly less than~$\mfrac{1}{2}$. 
For example 3SAT is in~\textsf{PP}: guess an assignment of the $n$ variables and if it
works, reports SAT; otherwise report UNSAT with probability $\mfrac{1}{2} + \mfrac{1}{2^{n+1}}$.
The decision problem \textbf{MAJ-3SAT} asks
`do the majority of assignments satisfy the 3SAT clauses': it is unlikely to have a polynomial
time certificate. Decision problems based on sampling DFT coefficients,
such as Deutsch--Jozsa and Simon's Problem are in $\mathsf{BQP}$, and (see the subsection
on oracles below) may not in $\mathsf{P}$. 
It is known that $\mathsf{P} \subseteq \mathsf{NP} \subseteq \mathsf{PP} 
\subseteq \mathsf{PSPACE} \subseteq \mathsf{EXPTIME}$
and that at least one of these containments is proper.\label{fig:comp} }

\end{figure}

\subsubsection*{Oracles}
It is important to note that the quantum gate $U_f$
and the classical subroutine implementing $f$ are each given as a black box.
For the Deutsch--Jozsa problem, it is conceivable that, knowing how $f$ is implemented as a sequence of polynomially many
classical logic operations, one can quickly tell whether~$f$ is constant or balanced, without
having to evaluate~$f$ at all.\footnote{\textbf{Aside:} to make this seem less likely,
reflect that the hash function SHA-256 is implemented as a sequence of polynomially many simple operations,
but almost certainly you do not have a quick way to produce inputs whose outputs  end with lots of zeros.}
Thus all we have proved is that \textsf{EXACT}\textsf{QP}, and so \textsf{BQP}, properly contain~\textsf{P} \emph{when both are taken
relative to a polynomial time oracle}. 

\subsection{Unitary gates and the Schr{\"o}dinger equation${}^\star$}\label{subsec:whyUnitary}
We finish by motivating our assumption that quantum gates are unitary, and making the connection
with quantum mechanics as you probably were first taught it.

\subsubsection*{Why unitary gates}
Our statement of the Born rule (see Definition~\ref{defn:measureZ}) 
 strongly motivates the axiom that quantum states evolve in a way
that \emph{preserves their norm}. It is a nice exercise in sesquilinear algebra to
show that if~$\mathcal{K}$ is a 
 Hilbert space and $L : \mathcal{K} \rightarrow \mathcal{K}$
is a norm-preserving linear map then~$L$ is unitary.\footnote{\textbf{Outline proof:}
given any $v$, $w \in \mathcal{K}$ we have $|| v + w||^2 = ||v||^2 + ||w||^2 + 2 \Re \langle v | w \rangle$ and $|| L(v+w) ||^2 = ||Lv||^2 + ||Lw||^2 + 2\Re \langle Lv | Lw \rangle $. Since
$L$ is norm preserving, we have \smash{$\Re \langle v | w \rangle = \Re \langle Lv | Lw \rangle = \Re \langle v |
\overline{L}^\transpose L w \rangle$}, implying that \smash{$\Re \langle v | (\overline{L}^\transpose L - I)w
\rangle =0$}. Since $v$ was arbitrary, we may take \smash{$v = (\overline{L}^\transpose L - I)w$} to deduce
that \smash{$\overline{L}^\transpose L - I = 0$}, i.e.~$L$ is unitary.}
 So let's take this also as axiomatic.
Then we can suppose that a gate is a smooth family of unitary maps $t \mapsto U(t)$
for $0 \le t \le 1$, and that what we see in our quantum computer are the outputs $U(1)\ket{\psi(0)}$
where $\ket{\psi(0)}$ is our starting state, for instance a zero qubit.\footnote{\textbf{Oversimplification:} we ignore
that gates might use measurement `under the hood'. Measurement
is norm-preserving by the Born rule, but of course not unitary.}
By a familiar argument from naive Lie theory, we have $U(\epsilon) = I + \epsilon M + O(\epsilon^2)$
for some matrix $M$ which, since $U(\epsilon)$ is unitary, satisfies
\[		(I + \epsilon \overline{M})^\transpose (I + \epsilon M) = I + O(\epsilon^2). \]
Thus $\overline{M}^\transpose + M = 0$, and so $\overline{M}^\transpose = -M$. That is,
$M$ is anti-Hermitian. We may therefore write $M = -iH$ for some Hermitian $H$. We choose
the minus sign so that the punchline of this section does not need a sign change.

\subsubsection*{To the Schr{\"o}dinger equation}
Mathematically, it is appealing to suppose that $U(t)$ is in fact the one-parameter
subgroup 
\[ 		t \mapsto \exp (-iHt) \]
of the unitary group on $\mathcal{K}$. Supposing this is the case, we find that
a starting state $\ket{\psi(0)}$ evolves as
$\ket{\psi(t)} = \exp( -iHt) \ket{\psi(0)}$, and so satisfies the differential equation
\[ \frac{\mathrm{d} \ket{\psi(t)}}{\mathrm{d} t} = -iH\! \ket{\psi(t)}. \]
The equation above is, up to a missing reduced Planck's constant, the Schr{\"o}dinger equation
\[ i \hbar \frac{\mathrm{d} \ket{\psi(t)}}{\mathrm{d} t} = H \! \ket{\psi(t)}. \]
One can work out that $\hbar$ 
should be on the left in the Schr{\"o}dinger equation by remembering one other
formula using Planck's constant and then using
dimensional analysis on the Hamiltonian~$H$.\footnote{\textbf{Units:} 
you do not need quantum mechanics 
to derive that energy has units $\mathrm{kg}\hskip1pt \mathrm{m}^2\hskip1pt \mathrm{s}^{-2}$.
From the formula $E = \hbar \nu$ for the 
energy of a photon in terms of its angular frequency, 
we deduce that $\hbar$ has units $\mathrm{kg} \hskip1pt\mathrm{m}^2 \hskip1pt\mathrm{s}^{-1}$. 
(This formula is $E = h f$ in terms of the usual frequency $f = \nu/ 2\pi$.)
Then since $H$ is a Hamiltonian, with units of energy, dimensional consistency in 
Schr{\"o}dinger equation requires that $\hbar$ appears on the left
in $ i \hbar \frac{\mathrm{d} \ket{\psi(t)}}{\mathrm{d} t}$. {\bf To do: natural units}}
(Another heuristic is that the energy given by applying the $H$ operator is very small, whereas derivatives
may be large, so we should scale the derivative by the tiny $\hbar  \approx 1.05 \times 10^{-34} \mathrm{kg}\hskip1pt \mathrm{m}^2\hskip1pt \mathrm{s}^{-1}$.)
See \cite[\S 9.2]{SusskindFriedmanQM} for why we should
now take $H = -\frac{\hbar^2}{2m} \frac{\mathrm{d}^2}{\mathrm{d}x^2} + V(x)$ 
to model a non-relativistic particle of mass $m$ moving in a potential~$V$; substituting 
and replacing $\ket{\psi(t)}$ with a traditional wave function $\psi(t,x)$ gives
Schr{\"o}dinger equation in the form
\[ i \hbar \frac{\mathrm{d} \psi(t,x)}{\mathrm{d} t} = -\frac{\hbar^2}{2m} \frac{\mathrm{d}^2\psi(t,x)}{\mathrm{d}x^2}  + V(x)\psi(t,x) \]
you might recall from a traditional first course in quantum mechanics.\footnote{\textbf{Opinion:}
isn't it about time such courses entered the 21st Century? Is it the misguided belief that 
solving differential equations is `easy' whereas Hilbert spaces such as $\C^2$ and operators are `hard' that
stops this happening? Or perhaps Hilbert spaces are indeed `hard':
the Hilbert space relevant to the Schr{\"o}dinger equation above
is the space of time-dependent square-integrable functions
on $\R$, on which the unitary operator $e^{-i H t}$ is defined
on a dense subset. See \cite[Ch.~14]{Hall} for a rigorous account of the Stone--von Neumann theorem.
}
% Aside: cannot be correct because relativistic invariance under Lorentz transform requires
% same order in x and t.

\section{Too many qubits: quantum error correction}\label{sec:tooManyQubits}

This section is particularly idiosyncratic: we first develop just enough
of the theory of stabiliser codes to give a carefully (maybe laboriously) motivated
quantum code that will correct a single $X$-error,
but will not guard against general quantum errors. 
We then show our construction is a special case of the key quantum technique
of `measuring a stabiliser' and use this to define error correction for the $[[7,1,3]]$-Steane code.

\subsection{Motivation: No cloning theorem}\label{subsec:noCloning}
In classical error correction we can discover the error,
provided it is not \emph{too} high weight,
by direct examination of the data.
The error is then a classical vector of bits. 
In quantum error correction,
the data is a quantum state,
and to avoid decohering it, we shall see that
we must take 
great care to learn only the syndrome of the error, and not the error itself. The error is resolved in the course of measurement 
to a combination of Pauli operators.
We make this precise in \S\ref{subsec:measuringStabiliser}.
And, since the map $\mathcal{H} \rightarrow \mathcal{H} \otimes \mathcal{H}$
defined by $\ket{\psi} \mapsto \ket{\psi}\ket{\psi}$ is not even linear, and, as defined here, not invertible,
there is no unitary map that 
we could use to back up quantum data. This is a simple version of the `No cloning theorem':
see \cite[\S 1.3.5]{NielsenChuang}. 

\subsection{Measurement for quantum error correction}
For quantum error correction one often needs to measure a subset of the qubits.
Definition~\ref{defn:measureZall} (measurement of all qubits in the $Z$-basis)
generalizes as follows.

\begin{definition}\label{defn:measureZpartial}[$Z$-basis measurement of some qubits]
Let
\[ \ket{\psi} = \sum_{v \in \F_2^m} \ket{\psi_v} \otimes \hskip1pt a_v\!\ket{v} \in \HH^{\otimes n} \otimes \HH^{\otimes m}
\]
 be a normalized state in which $\braket{\psi_v}{\psi_v} = 1$ for all $v \in \F_2^m$.
 \emph{Measuring $\ket{\psi}$ on the final $m$ qubits in the $Z$-basis}
projects it to $\ket{\psi_v}$ with probability $|a_v|^2$. The measurement result is $v$.
\end{definition}

For example measuring the Bell state \smash{$\mfrac{1}{\sqrt{2}} \bigl( \ket{00} + \ket{11} \bigr)$} 
in Exercise~\ref{ex:makeBell} on its second
qubit projects it to $\ket{00}$ and $\ket{11}$ with equal probability;
the measurement results are $0$ and $1$, respectively. Depending on your tolerance for the joke,
one last traditional pause might now be observed.
%For one final time, let us observe the traditional pause to note that \emph{measurement changes the state}.

\subsection{A toy code}\label{subsec:toyCode}
The classical
length three repetition code $\{000, 111\} \subseteq \F_2^3$ can correct a single bit flip. Motivated
by this, we might hope that the subspace $\CC$ of $\HH^{\otimes 3}$ spanned by $\ket{000}$ and $\ket{111}$
can be used as a quantum code to correct a single $X$-error. Here, by `used as a quantum code',
we mean that a general bare qubit $\alpha\! \ket{0} + \beta\! \ket{1}$ is \emph{encoded} in $\CC$ 
as the \emph{data state}
\[ \alpha\! \ket{000} + \beta\! \ket{111} \]
and we have available a set of \emph{logical gates} that implement the usual single qubit gates $X$, $Z$
and $H$ directly on the encoded qubits.\footnote{\textbf{Logical operations:} 
the transverse operator $X \otimes X \otimes X$ swaps $\ket{000}$ and $\ket{111}$
and so implements logical $X$, and either $Z \otimes I \otimes I$ or the transverse $Z \otimes Z \otimes Z$
fixes $\ket{000}$ and
flips the sign on $\ket{111}$, so implements logical $Z$. One limitation of this code
is shown by the difficulty of implementing logical Hadamard. The circuit
below implements logical
CNOT between two different code blocks; the example has input 
\smash{$\mfrac{1}{\sqrt{2}} \bigl(\ket{000} + \ket{111}) \bigr) \ket{000}$}
which is the tensor product of the logical plus \smash{$\mfrac{1}{\sqrt{2}} \bigl(\ket{000} + \ket{111})$}
 and logical zero states $\ket{000}$; the output is the logical Bell state.
\[
\hspace*{1.25in}\Qcircuit @R=15pt @C=12pt{   & \ctrl{3} & \qw & \qw & \qw \\ 
\llap{$\mfrac{1}{\sqrt{2}} \bigl( \ket{000} + \ket{111}\bigr)
\left\{
\phantom{ \begin{matrix} x \\ x \\ x \end{matrix}}\right.
 $}& \qw & \ctrl{3} & \qw& \qw  \\ & \qw & \qw & \ctrl{3} & \qw \\
  \llap{$\ket{0}$\ } & \targ{}  & \qw & \qw & \qw \\ \llap{$\ket{0}$\ }& \qw & \targ{}& \qw   & \qw \\ 
  \llap{$\ket{0}$\ }& \qw & \qw & \targ{}& \qw 
} \raisebox{-50pt}{$
\left. \phantom{ \begin{matrix}  x \\ x \\ x \\ x \\ x \\ x  \\ x \\ x \\ x\end{matrix}} \right\} 
\mfrac{1}{\sqrt{2}} \bigl( \ket{000}\ket{000} + \ket{111}\ket{111}\bigr)$}
\]
}
Decoding an output state known to be either $\ket{000}$ or $\ket{111}$ to a single classical bit 
is easy: just measure every qubit in the $Z$-basis. 
But this also shows the basic problem in quantum error correction.

\begin{exercise}
Suppose that the data state $\alpha \!\ket{000} + \beta \!\ket{111}$ 
has an $X$-error on the first qubit, and so it is
\[ (X \otimes I \otimes I) \bigl( \alpha \!\ket{000} + \beta \!\ket{111} \bigr)  = 
 \alpha\!\ket{100} + \beta\! \ket{011} .\]
Arguing that in a healthy data state, the first two qubits are equal (in each $Z$-basis summand),
we might decide to measure 
the first two qubits in the $Z$-basis.
Is this a good plan?\footnote{\textbf{Solution:} by
 Definition~\ref{defn:measureZpartial}, the measurement projects
 the data state 
to either $\ket{000}$ or $\ket{111}$ if there is no $X$-error,
or to either $\ket{100}$ or $\ket{011}$ if the error is $X \otimes I \otimes I$,
and the measurement result is $00$, $11$ or $10$, $01$, respectively, so we know which case applies.
The operation is a success, but the patient dies: we have destroyed the superposition
and lost the information in the coefficients $\alpha$ and $\beta$. This is \emph{not}
a successful decoding algorithm.} 
\end{exercise}

\subsubsection*{Measuring onto an ancilla}
A key idea in quantum error correction is that rather than directly
measuring the data state\footnote{\textbf{Trigger warning:} even the thought of this action makes the author feel uneasy.}, we first interact the data state with
a suitable ancilla, and then measure the ancilla. 
In this case a suitable ancilla is the zero qubit $\ket{0}$ and the necessary interaction can be done
using CNOT gates. 
As an indication of how to perform this interaction, the diagram below shows a plausible but
wrong circuit. Two possible input/output pairs are shown: the top is clean data, the bottom
has an $X$-error on the first qubit copied to the ancilla.
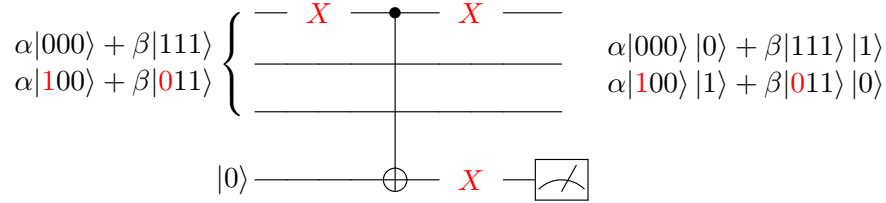
\[ \hspace*{1.5in} \Qcircuit@C=12pt @R=18pt{ & \qw & \red{X} & & \ctrl{3} & \qw &  \red{X} &  & \qw  \\
  \llap{$  \begin{matrix}  \alpha\!\ket{000} + \beta\!\ket{111} \\  \alpha\!\ket{{\red{1}}00} + \beta\!\ket{{\red{0}}11} \end{matrix} \left\{
 \phantom{ \begin{matrix} x \\ x \\ x \end{matrix}}\right.\!\!$}
 	  & \qw    & \qw & \qw & \qw  & \qw & \qw& \qw & \qw \\
 & \qw & \qw    & \qw  & \qw  & \qw & \qw & \qw & \qw\\
   \llap{$\ket{0}$\,} 
   & \qw    & \qw  & \qw   & \targ{}  & \qw & \red{X} &  & \meter } \ \raisebox{-20pt}{
 $  \begin{matrix} \alpha\!\ket{000} \ket{0} + \beta\!\ket{111}\ket{1} \\
   \alpha\!\ket{{\red 1}00} \ket{1} + \beta\!\ket{{\red 0}11}\ket{0} \end{matrix}$
   } \]  
By Definition~\ref{defn:measureZpartial}, measuring the ancilla collapses the data state: for instance
if $0$ is the measurement result that the new state is $\ket{000}\ket{0}$ if there is no error,
or $\ket{{\red 0}11} \ket{0}$ if there is an error.   
The problem is that the ancilla \emph{learns about the data state}: the measurement result determines
not just whether or not there is an $X$-error on one of the first two qubits, 
but also distinguishes the two logical codewords
$\ket{000}$ and $\ket{111}$. What if instead we
could somehow learn just about the error, and nothing about the data? (This is your final chance
to discover the key circuit for yourself.)
% using our earlier
%observation that in a healthy data state, the first two qubits are equal (in each $Z$-basis summand)?
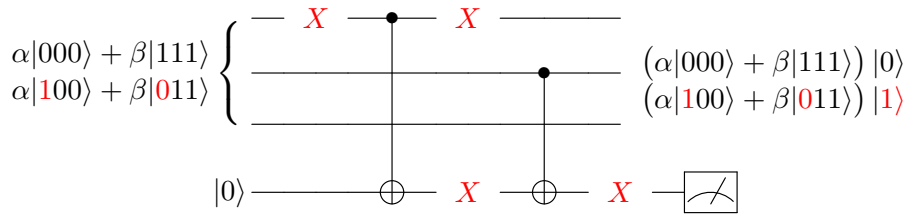
\[ \hspace*{1.5in} \Qcircuit@C=12pt @R=18pt{ & \qw & \red{X} & & \ctrl{3} &  \qw &  \red{X} &  & \qw & \qw & \qw \\
  \llap{$  \begin{matrix}  \alpha\!\ket{000} + \beta\!\ket{111} \\  \alpha\!\ket{{\red{1}}00} + \beta\!\ket{{\red{0}}11} \end{matrix} \left\{
 \phantom{ \begin{matrix} x \\ x \\ x \end{matrix}}\right.\!\!$}
 	  & \qw    & \qw & \qw &  \qw & \qw & \qw & \qw & \ctrl{2} & \qw & \qw&  \\
 & \qw & \qw    & \qw  & \qw  & \qw & \qw & \qw & \qw & \qw & \qw \\
   \llap{$\ket{0}$\,} 
   & \qw    & \qw  & \qw   & \targ{}  & \qw & \red{X} & & \targ{} & \qw &  \red{X} &  & \meter } \ \raisebox{-24pt}{
 $ \hspace*{-0.6in} \begin{matrix} \bigl(\alpha\!\ket{000}  + \beta\!\ket{111}\bigr)\ket{0} \\
   \bigl(\alpha\!\ket{{\red 1}00} + \beta\!\ket{{\red 0}11}\bigr)\ket{\red 1} \end{matrix}$
   } \]  

In this modified circuit we also CNOT from the second data qubit to the ancilla. 
This mirrors the parity check equation $x_1 + x_2 = 0$ satisfied by the classical codewords $000$ and $111$.
In both cases, clean data and $X$-error, the output
state is a pure tensor product of the input data state (with its $X$-error if relevant)
and $\ket{b}$, where $b={\red 1}$ if there is a single $X$-error on either of the first two qubits, and $b=0$
otherwise. We may therefore measure the ancilla qubit without collapsing the data state.
By Definition~\ref{defn:measureZpartial}, the measurement result is $0$ or~${\red 1}$ \emph{and the data state
is unchanged}.

\subsubsection*{Full syndrome information on $X$-errors for the toy code}
Generally we may suppose (see \S\ref{subsec:stochastic} below) that
the $X$-error is $X^{e_1} \otimes X^{e_2} \otimes X^{e_3}$ for some $e_1$, $e_2$, $e_3 \in \F_2$.
For instance, a single $X$-error on the first qubit is  $e = 100$.
The circuit just defined (repeated in the margin) has measurement result $e_1 + e_2$. This is the syndrome 
for the first row of the parity check matrix
\marginpar{\raisebox{48pt}{$ \quad \Qcircuit @C=12pt @R=12pt { & \ctrl{3} & \qw  & \qw \\ & \qw & \ctrl{2} & \qw  \\ & \qw & \qw  & \qw
\\ \llap{$\ket{0}$\,} & \targ{}
& \targ{} &  \meter }$} }
\[ P = \left( \begin{matrix}1 & 1 & 0 \\ 0 & 1 & 1 \end{matrix} \right) \]
of the classical repetition code of length $3$. Correspondingly $Z \otimes Z  \otimes I$
is a stabiliser of our code: since $(Z\otimes Z \otimes I) \ket{000} = \ket{000}$ and $(Z\otimes Z \otimes I) \ket{111}
= \ket{111}$, we have $(Z \otimes Z \otimes I) \ket{\psi} = \ket{\psi}$ for all $\ket{\psi} \in \CC$.
A similar circuit (see the margin) in which the CNOTs
are from the second and third data qubits \marginpar{\raisebox{12pt}{$\quad \Qcircuit @C=12pt @R=12pt { & \qw & \qw  & \qw \\ & \ctrl{2} & \qw & \qw  \\ & \qw & \ctrl{1} & \qw
\\ \llap{$\ket{0}$\,} & \targ{}
& \targ{} &  \meter }$} }
measures the syndrome  for the second row. Correspondingly, $I \otimes Z \otimes Z$ is a stabiliser of $\CC$.
Therefore we
may learn $e_1 + e_2$ and $e_2 + e_3$ which, since the classical repetition code of length $3$
is $1$-error correcting, is sufficient to determine the position of any single $X$-error.

\begin{exercise}
If the two measurement results are $(1,1)$, what single gate will correct the data state?
Assume that at most one $X$-error has occurred.\footnote{\textbf{Solution:} the unique solution to $e_1 + e_2 = 1$ and $e_2 + e_3 = 1$ having $e \in \F_2^3$
of weight $1$ is $e = 010$. We therefore correct by applying an $X$-gate to the second qubit.} 
\end{exercise}

An easy generalization of this exercise shows that the toy code can correct an arbitrary single $X$-error. 
%Well done for making this far. You might prefer to stop here, or skip to \S\ref{subsec:Steane}
%about the Steane code.

\subsection{Measuring a stabiliser}\label{subsec:measuringStabiliser}
To put this in a more general framework we dualize, as in \S\ref{sec:manyQubits},
by applying Hadamard gates to everything in sight. By the
identity $(H \otimes H) \CNOT_{12} (H \otimes H) = \CNOT_{21}$, diagrammed below,
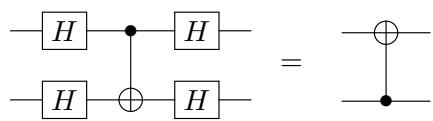
\[ \!\!\!\!\Qcircuit @C=12pt @R=12pt { & \gate{H} & \ctrl{1} & \gate{H} & \qw \\ & \gate{H} & \targ{} & \gate{H} & \qw }
 \raisebox{-12pt}{\quad =\quad\ } \raisebox{2pt}{\Qcircuit@C=12pt @R=20pt { & \targ{} & \qw \\ & \ctrl{-1} & \qw}} \]
this swaps the control and target in every CNOT gate.
The \emph{dual toy code} has codewords $\ket{+}\ket{+}\ket{+}$ and $\ket{-}\ket{-}\ket{-}$
which we abbreviate as $\ket{\pl\pl\pl}$ and $\ket{\mi\mi\mi}$ The dualized measurement circuits below
now measure syndromes of $Z$-errors.
The final measurement is, thanks to the $H$ gate
on the ancilla, correctly in the $Z$-basis; compare the two circuits ending \S\ref{subsec:QDFT}.
Note that by the copy rule for $Z$-faults, a single $Z$-fault on the data will be copied to the ancilla
wires, swapped to an $X$-fault by the Hadamard gates, and then recorded by the measurement results.
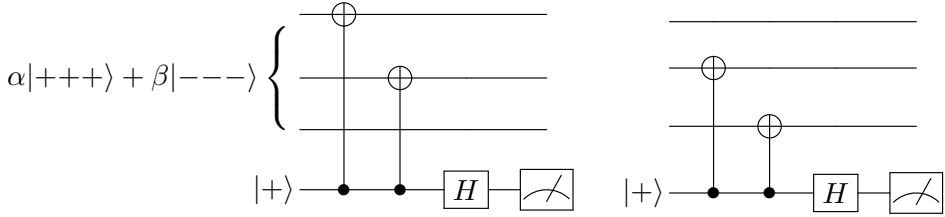
\[ \hspace*{0.5in} \Qcircuit @C=12pt @R=15pt { & \targ{} & \qw  & \qw & \qw \\ 
  \llap{$   \alpha\!\ket{\pl\pl\pl} + \beta\!\ket{\mi\mi\mi}  \left\{
 \phantom{ \begin{matrix} x \\ x \\ x \end{matrix}}\right.\!\!$}
& \qw & \targ{} & \qw& \qw   \\ & \qw & \qw  & \qw
& \qw \\ \llap{$\ket{+}$\,} & \ctrl{-3} 
& \ctrl{-2} &  \gate{H} & \meter } \hspace*{0.5in}
\raisebox{0pt}{\Qcircuit @C=12pt @R=13pt { & \qw & \qw  & \qw & \qw \\ 
%  \llap{$   \alpha\!\ket{\pl\pl\pl} + \beta\!\ket{\mi\mi\mi}  \left\{
% \phantom{ \begin{matrix} x \\ x \\ x \end{matrix}}\right.\!\!$}
& \targ{} & \qw& \qw & \qw   
\\ & \qw & \targ{}  & \qw
& \qw \\ \llap{$\ket{+}$\,} & \ctrl{-2} 
& \ctrl{-1} &  \gate{H} & \meter }}
\]
\begin{exercise}\label{ex:logicalIdentity}
The stabilisers found above of $\CC$ become stabilisers $X \otimes X \otimes I$
and $I \otimes X \otimes X$ of the dual code.
Using this, show that in the absence of errors each circuit implements the identity operation.\footnote{\textbf{Solution:} 
in the first circuit, the controlled CNOTs specify $X \otimes X \otimes I$ on
the data state, which is a stabiliser of the dual code. 
Thus, by a calculation similar
to footnote~\ref{footnote:CNOTlinear} or the following Example~\ref{ex:Zfault}, the controlled operation is the identity.  (That $X \otimes X \otimes I$ \emph{is}
a stabiliser can be seen without going through
the duality by noting that since $X \ket{+} = \ket{+}$, $X \ket{-} = -\ket{-}$, an arbitrary
data state $\alpha\!\ket{\pl\pl\pl} + \beta\!\ket{\mi\mi\mi}$ is stabilised by this circuit. 
The proof
for the second circuit is very similar.}
\end{exercise}

This exercise is the rigorous rephrasing of the property we saw is essential, that the ancilla
should `learn nothing about the data state'.

\begin{example}\label{ex:Zfault}
Suppose that 
there is a $Z$-error on the first qubit. Thus, since $Z \ket{+} = \ket{-}$, the data state
is $\ket{\psi} = \alpha\!\ket{\mi\pl\pl} + \beta\!\ket{\pl\mi\mi}$. 
To find the output of the first dualized circuit we 
could use the previous Exercise~\ref{ex:logicalIdentity} and the copy rule for $Z$-faults.
 Instead, to motivate Lemma~\ref{lemma:measureStabiliser},
we calculate
\begin{align*} 
\CNOT_{41} &\CNOT_{42}  \bigl( \alpha\!\ket{\mi\pl\pl} + \beta\!\ket{\pl\mi\mi}\bigr)
\ket{+}  \\
&= \CNOT_{41} \CNOT_{42} \bigl( \alpha\!\ket{\mi\pl\pl} + \beta\!\ket{\pl\mi\mi}\bigr)
\mfrac{1}{\sqrt{2}} \bigl( \ket{0} + \ket{1} \bigr) \\
&= \bigl( \alpha\!\ket{\mi\pl\pl} + \beta\!\ket{\pl\mi\mi}\bigr)
\mfrac{1}{\sqrt{2}} \ket{0} 
- \bigl( \alpha\!\ket{\mi\pl\pl} + \beta\!\ket{\pl\mi\mi}\bigr)
\mfrac{1}{\sqrt{2}} \ket{1} \\
&= \bigl( \alpha\!\ket{\mi\pl\pl} + \beta\!\ket{\pl\mi\mi}\bigr)
\mfrac{1}{\sqrt{2}}\bigl(  \ket{0} - \ket{1}  \bigr) \\
&= \bigl( \alpha\!\ket{\mi\pl\pl} + \beta\!\ket{\pl\mi\mi}\bigr)
\ket{-} \\
&= \ket{\psi}\ket{-}
\end{align*}
and deduce that, since $H\ket{-} = \ket{1}$,
the measurement result is $1$, with output data state $ \alpha\!\ket{\mi\pl\pl} + \beta\!\ket{\pl\mi\mi}$ unchanged
from the input. The argument for the second circuit is very similar, but in this case,
because $I \otimes X \otimes X$ commutes with $Z \otimes I \otimes I$, there is no $Z$-fault copied 
to the ancilla (see also Exercise~\ref{ex:copyUp} below; this
copying up is called `phase kick-back' in the setting of Lemma~\ref{lemma:measureStabiliser}), and so the measurement result is~$0$. Again the output data state is unchanged
from the input.
\end{example}

Observe that each $Z$-syndrome measurement
circuit is 
an instance of the general \emph{stabiliser measurement circuit}
shown  below, in which $U$ is a unitary operator on $\HH^{\otimes n}$
controlled on the bottom qubit. Exercise~\ref{ex:logicalIdentity} and Example~\ref{ex:Zfault}
 generalizes as follows.

\begin{lemma}[Measuring a stabiliser]\label{lemma:measureStabiliser}
Let $U : \HH^{\otimes n} \rightarrow \HH^{\otimes n}$ be a unitary map with eigenvalues $+1$ and $-1$.
Let $\ket{\psi} \in \HH^{\otimes n}$ be a normalized
state with unique expression $\ket{\psi} = \ket{\psi_0} + \ket{\psi_1}$
as a 
linear combination of a $+1$-eigenvector of $U$ and a $-1$-eigenvector of $U$.
The output of the circuit below
is $\ket{\psi_0}\ket{0}$ or $\ket{\psi_1}\ket{1}$
with probabilities $\braket{\psi_0}{\psi_0}$ or $\braket{\psi_1}{\psi_1}$ respectively,
and measurement results $0$ or $1$ respectively.
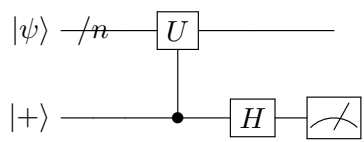
\[ 
 \Qcircuit @C=12pt @R=18pt {
   \llap{$\ket{\psi}$\ }  & \qw/n & \qw   & \gate{U} &  \qw &  \qw  \\
\llap{$\ket{+}$\ } & \qw   & \qw  & \ctrl{-1} &  \gate{H} &  \meter{} & %\cwx & 
} \]
\end{lemma}

\begin{proof}
Suppose first of all that $\ket{\psi}$ is an eigenstate with eigenvalue $\pm 1$.
Then the output before the $H$ gate on the control qubit or measurement is
\[ \begin{split}
\mfrac{1}{\sqrt{2}} \ket{\psi}\ket{0} + \mfrac{1}{\sqrt{2}} U \ket{\psi}\ket{1} 
&= \mfrac{1}{\sqrt{2}} \ket{\psi}\ket{0} \pm \mfrac{1}{\sqrt{2}} U \ket{\psi}\ket{1} \\
 &\qquad = \mfrac{1}{\sqrt{2}} \ket{\psi} \bigl( \ket{0} \pm \ket{1} \bigr) 
= \ket{\psi} \ket{\pm}.
\end{split}\]
Since $H\ket{+} = \ket{0}$ and $H\ket{-} = \ket{1}$ it follows by linearity
that the state created by the circuit just before measurement 
on a general $\ket{\psi}$ is $\ket{\psi_0}\ket{0} + \ket{\psi_1}\ket{1}$.
Now by Definition~\ref{defn:measureZpartial}, measuring the final
qubit projects the state to either $\ket{\psi_0}$ or $\ket{\psi_1}$ with the claimed probabilities.
\end{proof}

\begin{remark}
Thus the stabiliser measurement circuit is the identity 
in the case where $\ket{\psi}$ is a stabiliser state of $U$,
i.e. $U\ket{\psi} = \ket{\psi}$. The name `stabiliser measurement'
is perhaps misleading, but is standard.
\end{remark}

The stabiliser measurement circuit
can be used to perform all of quantum error correction.
For instance, we just saw in Example~\ref{ex:Zfault}
how to use it with $X$-gates in place for $U$ (so the controlled gate is a composition
of CNOTs) to get syndrome information about $Z$-errors on the dual toy code: specifically, 
the solution to Exercise~\ref{ex:Zfault} shows
that $Z^{e_1} \bigl( \alpha\!\ket{\mi\pl\pl} + \beta\!\ket{\pl\mi\mi}\bigr)$ is an eigenstate of $X\otimes X \otimes I$ with eigenvalue $(-1)^{e_1}$
and of $I \otimes X \otimes X$ with eigenvalue $1$. Observe that these stabilisers correspond to the
rows of the parity check matrix $P$ of the classical repetition code of length $3$ shown\marginpar{$\qquad\quad\left ( \begin{matrix} 1 & 1 & 0 \\ 0 & 1 & 1 \end{matrix} \right) $}
in \S\ref{subsec:toyCode} and repeated in the margin.\footnote{\textbf{How it works for the toy code:} to use
the stabiliser measurement circuit to get syndrome information about
$X$-faults on the original toy code, we need to measure a $Z$-stabiliser, using gates controlled by the ancilla qubit. 
This can be done using $CZ$-gates.
In \S\ref{subsec:toyCode}, since the aim was to do everything using CNOT gates,
we instead measured the $Z$-stabiliser in a way that is less obvious, using CNOT gates 
controlled by the \emph{data} qubits. But the effect is the same. In \cite[Figure~10.15]{NielsenChuang}
this is called a `useful implication'.} As a further demonstration, in \S\ref{subsec:Steane}, we use the syndrome measurement circuit to show how to decode the Steane code and to justify our implicit assumption that 
every quantum error is a combination of the Pauli $X$- or $Z$-operators.

\begin{remark}[One circuit is all you need]
By Lemma~\ref{lemma:measureStabiliser}, and as just seen from  Exercise~\ref{ex:Zfault}, 
the stabiliser measurement circuit performs
phase estimation for eigenstates with $\pm 1$ eigenvalues. It can be generalized to perform
arbitrary phase estimation and hence, allowing ancillae on multiple qubits, the 
Quantum Discrete Fourier Transform on $\Z/2^n \Z$. A taste of this is given by the
circuits for the Quantum Discrete Fourier Transform on $\F_2^n$ at the  end of \S\ref{subsec:QDFT}; note they are upside down for this
section, having controls on the \emph{top} wires.
\end{remark}

%\begin{exercise}
%If we measurement result is ${\red 1}$, should we aim to correct the first qubit?\footnote{\textbf{Solution:}
%not necessarily, because it could be that the single $X$-error is on the second qubit.}
%\end{exercise}

\subsection{Fun with stabilisers${}^\star$} \label{subsec:GHZ}
Here we use stabilisers to show that a 
class of `locally realistic' interpretations of quantum theory are fundamentally flawed.
Please skip to \S\ref{subsec:Steane} if you want to get straight to the Steane code.
Let $\ket{\chi} = \mfrac{1}{\sqrt{2}} \bigl( \ket{000} + \ket{111} \bigr)$ be the \emph{cat state}
on three qubits.\footnote{\textbf{Etymology:} the cat state is named after Schr{\"o}dinger's cat: by
Definition~\ref{defn:measureZpartial}, 
when any single qubit is measured in the $Z$-basis, the state collapses to either $\ket{000}$ or $\ket{111}$} 
Since $\ket{\chi}$ lies in the toy code it is stabilised 
by $Z \otimes Z \otimes I$ and $I \otimes Z \otimes Z$. %No room for of course can be seen directly
Clearly $X \otimes X \otimes X$
is also a stabiliser. Since $Y = iXZ$, it follows that 
\[ (X \otimes X \otimes X)(I \otimes Z \otimes Z) = X \otimes \mfrac{1}{i} Y \otimes \mfrac{1}{i} Y = -X \otimes Y \otimes Y \]
is a stabiliser. Noting the minus sign on the right-hand side, we deduce that $\ket{\chi}$ 
is a $-1$-eigenvector for $X \otimes Y \otimes Y$.
Physically, this means that if three qubits are entangled in the cat state, and then
measured in the $x$-, $y$- and $y$-directions, the product of the measured eigenvalues is $-1$.
In a strong version of a 
\emph{locally realistic} interpretation of quantum theory, there
are hidden variables $x_1$, $x_2$, $x_3$, $y_1$, $y_2$, $y_3$ which are the deterministic, `known to God',
results of measuring the $3$ qubits in the $x$- and $y$-directions.
Thus $x_1 y_2y_3 = -1$. By symmetry, we also have $y_1x_2y_3 = -1$ and $y_1y_2x_3 = -1$.
Hence
\[ x_1x_2x_3 = (x_1y_2y_3)(y_1x_2y_3)(y_1y_2x_3) = (-1)^3 = -1. \]
According to the locally realistic interpretation, this means that the result of measuring
the $3$-qubits all in the $x$-direction is $-1$. But since $X \otimes X \otimes X$ is a stabiliser
of $\ket{\chi}$, in fact the result is $1$.\footnote{\textbf{Further remarks:}
This \emph{GHZ-experiment} is an improvement on the earlier `Bell's paradox' experiment using two qubits
entangled in the Bell state.
This
demonstrates a higher correlation between single qubit measurements that is possible with classical
physics, but it is still conceivable (just overwhelmingly improbable) that a locally realistic theory
gives the same results in repeated experiments. 
In contrast, the GHZ-experiment is `one-shot': it has been performed in
the laboratory using spatially separated qubits (so the hidden variables for qubit $2$
cannot `update' themselves after qubit $1$ is measured, but just before qubit $2$ is measured, without
breaking causality). The account above is a simplified version of \cite{MerminWhatsWrong}.
See also \cite[Experiment 8, page 29]{Maudlin}. Incidentally,
Maudlin is clear that my interpretation of
`local realism' is unhelpful and too strong: `Bell
proved that no local theory, full stop, can predict violations of his inequality. \ldots
If I had my druthers ``realist'' and ``anti-realist'' would be banned from these foundational
discussions.' \cite[page~xiii]{Maudlin}}

\subsection{The Steane $[[7,1,3]]$-code} \label{subsec:Steane}
In \S\ref{subsec:toyCode} we made a quantum code whose $Z$-stabilisers
$I \otimes I \otimes I$, $Z \otimes Z \otimes I$, $Z \otimes I \otimes Z$,
$I \otimes Z \otimes Z$ corresponding to rows in the row-space of the parity
check matrix $P$ of the repetition code of length $3$. 
The dual code has $X$-stabilisers obtained by replacing each $Z$ with $X$, and these were what we measured
in the previous subsection. 

\subsubsection*{A stabiliser code}
We now copy this strategy replacing $P$ 
with the parity check matrix of the Hamming $[[7,4,3]]$-code below, chosen
so that for each $i \in \{1,\ldots, 7\}$, column $i$ is the binary form of $i$.
\[ P = \left( \begin{matrix} 1 & 0 & 1 & 0 & 1 & 0 & 1\\ 0 & 1 & 1 & 0 & 0 & 1 & 1 \\ 0 & 0 & 0 & 1 & 1 & 1 & 1
\end{matrix} \right) \]
Given $u \in \F_2^7$, let $X^u = X^{u_1} \otimes \cdots \otimes X^{u_7}$ and similarly
let $Z^w = Z^{w_1} \otimes \cdots \otimes Z^{w_7}$. The Steane code is then
the subspace $\CC$ of $\HH^{\otimes 7}$ of all vectors fixed by each $X^u$ 
and $Z^w$ for $u$, $w \in \langle P \rangle_\mathrm{row}$, the row span of $P$.
Define the \emph{zero-logical} and \emph{one-logical} states by
\begin{align*}
\ket{0}_L &= \frac{1}{\sqrt{2^3}}\sum_{v \in \langle P \rangle_\mathrm{row}} \ket{v} \\
\ket{1}_L &= \frac{1}{\sqrt{2^3}}\sum_{v \in \langle P \rangle_\mathrm{row}} \ket{\overline{v}} \end{align*}
where $\overline{v} = v + 1111111$ is the bit flip of $v$.
(The normalization factor is $\sqrt{2^3}$ because $P$ has three
linearly independent rows, and so $\langle P\rangle_\mathrm{row}$ has $2^3$ elements.)
Since $X^u \ket{v} = \ket{u+v}$ for all $u$, $v\in \F_2^7$,
it is clear that $\ket{0}_L$ and $\ket{1}_L$ are fixed by the $X$-stabiliser
subgroup. We leave it as an extended exercise using Lemma~\ref{lemma:transverseHadamard}
to show that $\ket{0}_L$ and $\ket{1}_L$ are also fixed by the $Z$-stabiliser group
and so $\ket{0}_L$, $\ket{1}_L \in \CC$. With a bit more work,
this approach shows that $\CC$ is exactly the $2$-dimensional subspace of $\HH^{\otimes 7}$
spanned by the logical codewords $\ket{0}_L$ and~$\ket{1}_L$.
Alternatively this follows from the general theory of CSS codes\footnote{\textbf{Reference:}
see \cite[\S 10.4.2]{NielsenChuang} or, for a mathematician friendly introduction to CSS
codes and stabiliser codes going into much more detail than these notes, \cite[Ch.~7]{Preskill}.}
noting that $\langle P \rangle^\perp_\mathrm{row} = \langle P \rangle_\mathrm{row} 
\oplus \langle 1111111 \rangle$ and so 
$\langle P \rangle^\perp_\mathrm{row} / \langle P \rangle_\mathrm{row}$
consists of the two cosets
 $\langle P \rangle_\mathrm{row}$ and $1111111 + \langle P \rangle_\mathrm{row}$. Thus $\CC$ uses $7$ physical qubits to encode
$1$ logical qubit. 

\subsubsection*{Syndrome extraction}
The  circuit below measures the $X$-stabiliser for each row of $P$ as in Lemma~\ref{lemma:measureStabiliser}
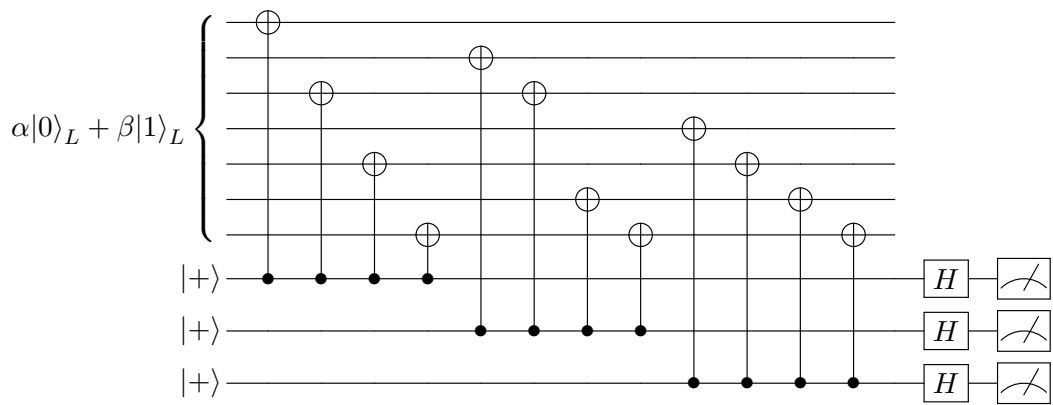
\[ \qquad\qquad\Qcircuit @C=1em @R=0.4em {
    % Data Wire 1
    & \targ & \qw   & \qw   & \qw   & \qw   & \qw   & \qw   & \qw   & \qw   & \qw   & \qw   & \qw   & \qw \\
    % Data Wire 2
    & \qw   & \qw   & \qw   & \qw   & \targ & \qw   & \qw   & \qw   & \qw   & \qw   & \qw   & \qw   & \qw \\
    % Data Wire 3
    & \qw   & \targ & \qw   & \qw   & \qw   & \targ & \qw   & \qw   & \qw   & \qw   & \qw   & \qw   & \qw \\
    % Data Wire 4
       \llap{$   \alpha\!\ket{0}_L + \beta\!\ket{1}_L  \left\{
 \phantom{ \begin{matrix} x \\ x \\ x \\ x \\ x \\ x\end{matrix}}\right.\!\!$}  & \qw   & \qw   & \qw   & \qw   & \qw   & \qw   & \qw   & \qw   & \targ & \qw   & \qw   & \qw   & \qw \\
    % Data Wire 5
    & \qw   & \qw   & \targ & \qw   & \qw   & \qw   & \qw   & \qw   & \qw   & \targ & \qw   & \qw   & \qw \\
    % Data Wire 6
    & \qw   & \qw   & \qw   & \qw   & \qw   & \qw   & \targ & \qw   & \qw   & \qw   & \targ & \qw   & \qw \\
    % Data Wire 7
    & \qw   & \qw   & \qw   & \targ & \qw   & \qw   & \qw   & \targ & \qw   & \qw   & \qw   & \targ & \qw \\
  \llap{$\ket{+}$\,}  & \ctrl{-7} & \ctrl{-5} & \ctrl{-3} & \ctrl{-1} & \qw   & \qw   & \qw   & \qw   & \qw   & \qw   & \qw   & \qw   & \qw & \gate{H} & \meter \\
  \llap{$\ket{+}$\,}  & \qw   & \qw   & \qw   & \qw   & \ctrl{-7} & \ctrl{-6} & \ctrl{-3} & \ctrl{-2} & \qw   & \qw   & \qw   & \qw   & \qw & \gate{H} & \meter \\
  \llap{$\ket{+}$\,}  & \qw   & \qw   & \qw   & \qw   & \qw   & \qw   & \qw   & \qw   & \ctrl{-6} & \ctrl{-5} & \ctrl{-4} & \ctrl{-3} & \qw & \gate{H} & \meter \\
}
\]
Dualizing this circuit (so going in the opposite direction to the passage from the toy code to the
dual toy code) we get a circuit as in \S\ref{subsec:toyCode}, still using only CNOT gates, that measures the $Z$-stabilisers.
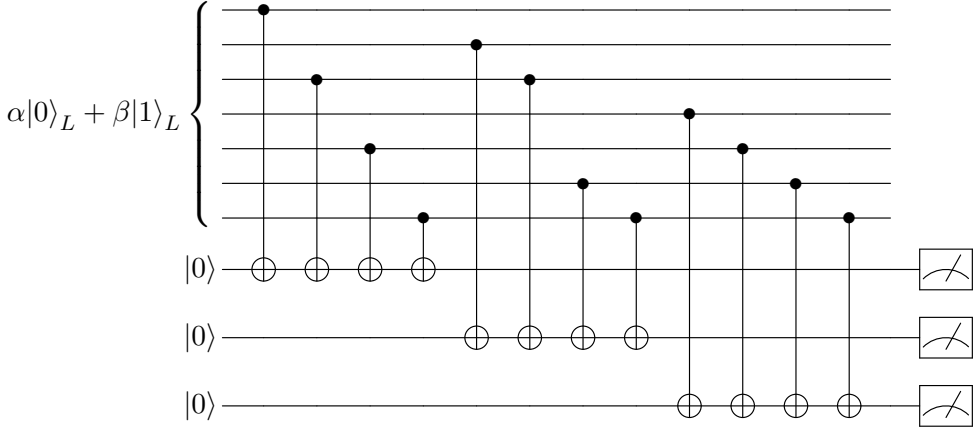
\[ \qquad\!\!\Qcircuit @C=1em @R=0.95em {
    % Data Wire 1
    & \ctrl{7} & \qw  & \qw  & \qw  & \qw  & \qw  & \qw  & \qw  & \qw  & \qw  & \qw  & \qw  & \qw \\
    % Data Wire 2
    & \qw  & \qw  & \qw  & \qw  & \ctrl{7} & \qw  & \qw  & \qw  & \qw  & \qw  & \qw  & \qw  & \qw \\
    % Data Wire 3
    & \qw  & \ctrl{5} & \qw  & \qw  & \qw  & \ctrl{6} & \qw  & \qw  & \qw  & \qw  & \qw  & \qw  & \qw \\
    % Data Wire 4
    \llap{$   \alpha\!\ket{0}_L + \beta\!\ket{1}_L  \left\{
 \phantom{ \begin{matrix} x \\ x \\ x \\ x \\ x \\ x\end{matrix}}\right.\!\!$}  & \qw  & \qw  & \qw  & \qw  & \qw  & \qw  & \qw  & \qw  & \ctrl{6} & \qw  & \qw  & \qw  & \qw \\
    % Data Wire 5
    & \qw  & \qw  & \ctrl{3} & \qw  & \qw  & \qw  & \qw  & \qw  & \qw  & \ctrl{5} & \qw  & \qw  & \qw \\
    % Data Wire 6
    & \qw  & \qw  & \qw  & \qw  & \qw  & \qw  & \ctrl{3} & \qw  & \qw  & \qw  & \ctrl{4} & \qw  & \qw \\
    % Data Wire 7
    & \qw  & \qw  & \qw  & \ctrl{1} & \qw  & \qw  & \qw  & \ctrl{2} & \qw  & \qw  & \qw  & \ctrl{3} & \qw \\
   \llap{$\ket{0}$\,}  & \targ    & \targ    & \targ    & \targ    & \qw  & \qw  & \qw  & \qw  & \qw  & \qw  & \qw  & \qw  & \qw & \meter \\
   \llap{$\ket{0}$\,}  & \qw  & \qw  & \qw  & \qw  & \targ    & \targ    & \targ    & \targ    & \qw  & \qw  & \qw  & \qw  & \qw & \meter \\
   \llap{$\ket{0}$\,}  & \qw  & \qw  & \qw  & \qw  & \qw  & \qw  & \qw  & \qw  & \targ    & \targ    & \targ    & \targ    & \qw & \meter \\
}
\]

\begin{exercise}
Suppose that the results of $X$-stabiliser measurement are $(1,0,1)$ and the results of $Z$-stabiliser
measurement are $(0,1,0)$. Assuming that at most
one $X$-error and at most one $Z$-error have occurred, what correction would you impose on the 
data state?\footnote{\textbf{Solution:} Since $(1,0,1)$ and $(0,1,0)$ are columns $5$ and $2$ of $P$ respectively,
and $X$-stabilisers give syndrome information about $Z$-errors (and \emph{vice-versa}), 
the correction is $X_2Z_5$.}
\end{exercise}

It is a routine generalization from this exercise to see that $\CC$ is $1$-error correcting.
It is therefore a $[[7, 1, 3]]$-quantum code.
In fact, writing $P_i$ for a Pauli operator on qubit $i$, where $P$ is either $X$ or $Z$,
there is a direct
sum $\mathcal{D} = \CC \oplus X_1 \CC \oplus \cdots \oplus X_7 \CC$
and then, by dimension counting using that $\dim \HH^{\otimes 7} = 2^7$ while $\dim \mathcal{D} = 8 \times 2
= 16 = 2^4$, we have
$\HH^{\otimes 7} = \mathcal{D} \oplus Z_1 \mathcal{D} \oplus \cdots \oplus Z_7 \mathcal{D}$:
the reason from quantum error correction that
the subspaces form a direct sum is because they have different syndromes. More mathematically,
it is because they are distinct joint eigenspaces for the $X$- and $Z$-stabiliser subgroups. 
 This is the quantum analogue
of the classical fact that the Hamming $[7,4,3]$-code is perfect, that is,
the Hamming balls of radius $1$ about codewords partition~$\F_2^7$.

\begin{exercise}\label{ex:copyUp}
Suppose that the data state is error free, but 
a $Z$-fault occurs on the first ancilla qubit 
after the first two CNOT gates. What are the $X$-stabiliser
and $Z$-stabiliser measurements (supposing we measure the 
$X$-stabilisers first)? What is the result of error correction?
Conclude that in this case `the cure was worse than the disease'.
\footnote{\textbf{Solution}: by the
copy rules in \S\ref{subsec:copyRules} the
$Z$-fault copies up to a $Z$-error $Z^{0000101} = I \otimes
I \otimes I \otimes I \otimes Z \otimes I \otimes Z$ on the data.
This would have been detected by the $X$-stabiliser measurements,
except that they have already been performed. Thus, since
the data state began error-free, all stabiliser measurements
are $0$ and no correction is imposed. But a single $Z$-fault
in the process became a weight $2$-error. Such `explosive' faults
mean that the naive version of Shor-style error correction 
in this section is impractical. There are ways around the problem,
at the cost of using more complicated ancillae, but they are beyond
the scope of these notes. See for instance Exercise 10.73 in \cite{NielsenChuang}.}
\end{exercise}

\subsubsection*{Making $\ket{0}_L$}
Recall from \S\ref{sec:twoQubits} that a CNOT gate is a device for creating entanglement:
for instance this was seen in Exercise~\ref{ex:makeBell}.
\marginpar{\setlength\arraycolsep{3pt}\scalebox{0.8}{$\left( \begin{matrix} \mathbf{1} & 0 & 1 & 0 & 1 & 0 & 1\\ 0 & \mathbf{1} & 1 & 0 & 0 & 1 & 1 \\ 0 & 0 & 0 & \mathbf{1} 
& 1 & 1 & 1 
\end{matrix} \right)$}}
The highlighted pivot columns in the Steane code $P$-matrix shows that qubit $1$ should be entangled with qubits $3,5,7$;
qubit $2$ should be entangled with qubits $3,6,7$ and qubit $4$ should be entangled with qubits $5,6,7$.
This should motivate the CNOT circuit in Figure~\ref{fig:SteaneZeroLogical}  that prepares $\ket{0}_L$ in the Steane code.

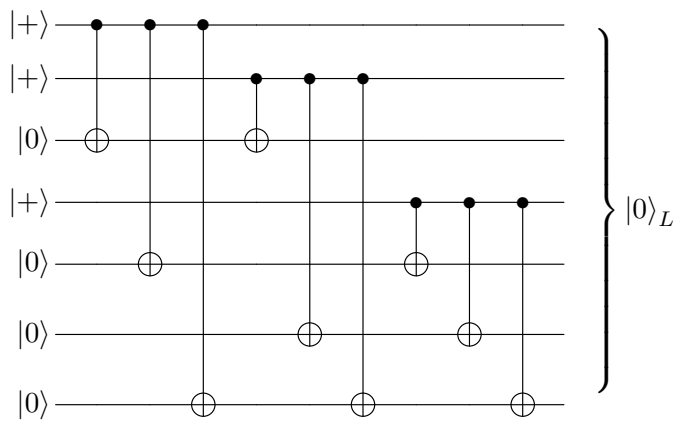
\begin{figure}[h!]
\[ \qquad\qquad\Qcircuit @C=1em @R=1.6em {
  \llap{$\ket{+}$\,} & \ctrl{2} & \ctrl{4} & \ctrl{6}  & \qw      & \qw      & \qw      & \qw      & \qw      & \qw      & \qw \\
  \llap{$\ket{+}$\,} & \qw      & \qw      & \qw       & \ctrl{1} & \ctrl{4} & \ctrl{5} & \qw      & \qw      & \qw      & \qw \\
  \llap{$\ket{0}$\,} & \targ{}  & \qw      & \qw       & \targ{}  & \qw      & \qw      & \qw      & \qw      & \qw      & \qw \\
  \llap{$\ket{+}$\,} & \qw      & \qw      & \qw       & \qw      & \qw      & \qw      & \ctrl{1} & \ctrl{2} & \ctrl{3} & \qw \\
  \llap{$\ket{0}$\,} & \qw      & \targ{}  & \qw       & \qw      & \qw      & \qw      & \targ{}  & \qw      & \qw      & \qw \\
  \llap{$\ket{0}$\,} & \qw      & \qw      & \qw       & \qw      & \targ{}  & \qw      & \qw      & \targ{}  & \qw      & \qw \\
  \llap{$\ket{0}$\,} & \qw      & \qw      & \targ{}   & \qw      & \qw      & \targ{}  & \qw      & \qw      & \targ{}  & \qw }
  \raisebox{-70pt}{$\ \left.\phantom{\begin{matrix}x \\ x \\ x \\ x \\ x \\ x \\ x \\ x \\ x \\ x \end{matrix}}\right\}\ket{0}_L$}
\]
\caption{Preparation circuit for zero logical $\ket{0}_L$ in the Steane code.\label{fig:SteaneZeroLogical}}
\end{figure}

The calculation below shows the effect
of each group of three CNOT gates on a $Z$-basis input state:
\begin{align*} \ket{b_1}\ket{b_2}\ket{0}\ket{b_4}&\ket{0}\ket{0}\ket{0} \\
&\mapsto 
\ket{b_1}\ket{b_2}\ket{b_1}\ket{b_4}\ket{b_1}\ket{0}\ket{b_1} \\
&\mapsto
\ket{b_1}\ket{b_2}\ket{b_1+b_2}\ket{b_4}\ket{b_1}\ket{b_2}\ket{b_1+b_2} \\
&\mapsto
\ket{b_1}\ket{b_2}\ket{b_1+b_2}\ket{b_4}\ket{b_1+b_4}\ket{b_2+b_4}\ket{b_1+b_2+b_4}.
\end{align*}
The result now follows by linearity (as in Exercise~\ref{ex:makeBell},
which is the analogous result for the smaller $P$-matrix $(1\,1)$), using that the input state is 
\[ \ket{+}\ket{+}\ket{0}\ket{+}\ket{0}\ket{0}\ket{0}
= \mfrac{1}{2^{3/2}}  \sum_{b_1,b_2,b_4 \in \{0,1\}} \ket{b_1}\ket{b_2}\ket{0}\ket{b_4}\ket{0}\ket{0}\ket{0},\]
and that, since columns $1$, $2$ and $4$ are pivot columns, a general element in $\langle P \rangle_\mathrm{row}$ is $(b_1,b_2,b_1+b_2,b_4,b_1+b_4,b_2+b_4,b_1+b_2+b_4)$.

\begin{exercise}
Use the fault pushing rules as in Exercise~\ref{ex:makeBellPushing} 
to give an alternative proof that the output of the circuit has $X$-stabiliser group given by the row span of $P$
and so is $\ket{0}_L$.
[\emph{Hint:} you are encouraged to use the shortcut that the output of any CNOT network is a CSS state, 
of the form $2^{-\dim M /2} \sum_{v \in M} \ket{v}$ for some matrix $M$.]\footnote{\textbf{Solution:}
$X$-faults on the qubits initialized $\ket{+}$ copy to the rows of the $P$-matrix;
thus the $X$-stabiliser group of the input, namely $\langle X_1, X_2, X_4 \rangle$ is conjugated
to $\bigl\langle X^u : u \in \langle P \rangle_\mathrm{row} \bigr\rangle$, which by the 
hint implies that the output state is $\ket{0}_L$.}
\end{exercise}

\begin{exercise}
Show that transverse $X$, i.e.~$X_1X_2X_3X_4X_5X_6X_7 = X \otimes X \otimes X \otimes X \otimes X \otimes X \otimes X$
is the $X$-logical operation for the Steane code that swaps $\ket{0}_L$ and $\ket{1}_L$,
and hence find a circuit preparing $\ket{1}_L$. Hence, or otherwise,
find a circuit preparing a general $\alpha\!\ket{0} + \beta\!\ket{1}$. [\emph{Hint:} 
conjugate transverse $X$ back through the circuit.]\footnote{\textbf{Solution:}
by definition 
\[ \begin{split}
\ket{1}_L \!=\! \sum_{v \in \langle P \rangle_\mathrm{row}} \!\!\ket{\overline{v}} =
 \sum_{v \in \langle P \rangle_\mathrm{row}}\!\! \ket{v+1111111} =
 X \otimes \cdots \otimes X
\sum_{v \in \langle P \rangle_\mathrm{row}} %(X \otimes X \otimes X \otimes X \otimes X \otimes X \otimes X) 
\ket{v} =  X \otimes \cdots \otimes X \ket{0}_L. \end{split} \]
Therefore applying a final transverse $X$-operation, i.e.~$X$ gates on all $7$ qubits, to the preparation
circuit for $\ket{0}_L$ makes $\ket{1}_L$. Conjugating this back through the circuit
by the copy rules in \S\ref{subsec:copyRules} gives $X_1X_2X_3X_4X_5X_6= X \otimes X \otimes X \otimes
X \otimes X \otimes X \otimes I$. It follows that the image of
$\ket{+}\ket{+}(\alpha\!\ket{0} + \beta\ket{1})\ket{+}\ket{0}\ket{0}\ket{0}$
under the composition of $\CNOT_{31}\CNOT_{32}\ldots \CNOT_{36}$ followed by the preparation
circuit making $\ket{0}_L$ is $\alpha\!\ket{0}_L + \beta\!\ket{1}_L$, as shown in the diagram below.
\[ \scalebox{0.7}{$\Qcircuit @C=1em @R=1.5em {
  \llap{$\ket{+}$\,} & \targ{}  & \qw & \qw & \qw &  \qw & \ctrl{2} & \ctrl{4} & \ctrl{6}  & \qw      & \qw      & \qw      & \qw      & \qw      & \qw      & \qw \\
  \llap{$\ket{+}$\,} & \qw &  \targ{}   & \qw & \qw& \qw   & \qw   & \qw      & \qw       & \ctrl{1} & \ctrl{4} & \ctrl{5} & \qw      & \qw      & \qw      & \qw \\
  \llap{$\alpha\!\ket{0}+\beta\!\ket{1}$\,} & \ctrl{-2} & \ctrl{-1} & \ctrl{1} & \ctrl{2} & \ctrl{3}  & \targ{}  & \qw      & \qw       & \targ{}  & \qw      & \qw      & \qw      & \qw      & \qw      & \qw \\
  \llap{$\ket{+}$\,} & \qw & \qw & \targ{}  & \qw & \qw & \qw      & \qw      & \qw       & \qw      & \qw      & \qw      & \ctrl{1} & \ctrl{2} & \ctrl{3} & \qw \\
  \llap{$\ket{0}$\,} & \qw & \qw & \qw & \targ{}  & \qw  & \qw     & \targ{}  & \qw       & \qw      & \qw      & \qw      & \targ{}  & \qw      & \qw      & \qw \\
  \llap{$\ket{0}$\,} & \qw & \qw & \qw & \qw & \targ{}  & \qw      & \qw      & \qw       & \qw      & \targ{}  & \qw      & \qw      & \targ{}  & \qw      & \qw \\
  \llap{$\ket{0}$\,} &\qw & \qw & \qw &  \qw & \qw  & \qw      & \qw      & \targ{}   & \qw      & \qw      & \targ{}  & \qw      & \qw      & \targ{}  & \qw }
  \raisebox{-64pt}{$\ \left.\phantom{\begin{matrix}x \\ x \\ x \\ x \\ x \\ x \\ x \\ x \\ x \\ x \\ x \\ x \end{matrix}}\right\}\alpha\!\ket{0}_L + \beta\!\ket{1}_L$}$}
\]

\smallskip\noindent
See \cite[\S 4.2]{Gottesman} for generalizations of this. Another important way to encode uses an auxiliary qubit
in state $\alpha\!\ket{0} + \beta\!\ket{1}$; a transverse CNOT with this qubit as the control with target $\ket{0}_L$ gives
the state $\alpha\!\ket{0}\ket{0}_L + \beta\!\ket{1}\ket{1}_L$; now apply a Hadamard gate to the top qubit
to get 
\[ \alpha\!\ket{+}\ket{0}_L + \beta\!\ket{-}\ket{1}_L = \mfrac{1}{\sqrt{2}}\ket{0}\bigl(\alpha\ket{0}_L + \beta\ket{1}_L\bigr)
+ \mfrac{1}{\sqrt{2}}\ket{1}\bigl(\alpha\ket{0}_L - \beta\ket{1}_L\bigr).\]
Thus measuring the top qubit in the $Z$-basis gives the required state if the measurement result is $0$.
If the measurement result is $1$ then transverse $Z$ must be applied.
}
\end{exercise}

One motivation for making $\ket{0}_L$ is that this
is the ancilla state for Steane style
error correction of $Z$-errors. For more on this
see \cite{Steane}.

\subsection{Justifying the stochastic $X$- and $Z$-error model}\label{subsec:stochastic}
We finish with an example justifying our implicit assumption
that quantum errors appear as Pauli operators on one or more  qubits.

\begin{example}
The data state in the Steane code was
\[ \ket{-}_L = \mfrac{1}{\sqrt{2}}\bigl( \ket{0}_L - \ket{1}_L\bigr)\]
until a rogue environmental electron got entangled with it, forming the new state on $8$ qubits
\[ \ket{\phi} = \mfrac{4}{5} \ket{0}_\mathrm{env}\ket{-}_L  + \mfrac{3}{5} X_2 \ket{1}_\mathrm{env}\ket{-}_L . \]
Here we use the notation just introduced, that $X_2$ is a Pauli $X$ on the second qubit.
Measuring the first $Z$-stabiliser $Z_1Z_3Z_5Z_7$, corresponding to the first row $1010101$ of $P$,
leaves~$\ket{\phi}$ unchanged.
But since $Z_2Z_3Z_6Z_7$, corresponding to the second row $0110011$ of $P$, anticommute with $X_2$,
the state in the first syndrome extraction circuit above \emph{just before measurement} is
\[ \mfrac{4}{5} \ket{0}_\mathrm{env} \ket{-}_L \ket{0} + \mfrac{3}{5} X_2\ket{1}_\mathrm{env} \ket{-}_L
\ket{1}. \]
By Definition~\ref{defn:measureZpartial}, measuring the ancilla qubit gives output
state $\ket{0}_\mathrm{env}\ket{-}_L$ with probability $\mfrac{16}{25}$ (result is $0$)
and output state $\ket{1}_\mathrm{env}X_2\ket{-}_L$ with probability $\mfrac{9}{25}$ (result is $1$). Note that the environment state is no longer entangled with the data state, and that after the final $Z$-stabiliser $Z_4Z_5Z_6Z_7$, corresponding
to the third row $0001111$ of $P$, is measured,  we, or rather
the quantum computer, now know what $X$-correction to impose on the data state.
\end{example}

More generally, one can write an arbitrary quantum error as a linear combination of the Pauli
matrices and use a generalization of this example to show that the process of stabiliser measurement
(or equivalently, syndrome extraction) forces the error to decohere into a specific combination
of Pauli operators.
Thus it is \emph{not} a law of nature that quantum errors manifest as $X$- and $Z$-faults on the data:
rather this is an emergent property of our decoding logic.\footnote{\textbf{In practice:} everyone
seems to pretend that nature is sufficiently obliging so as to create errors in the mathematically
convenient fashion, and thanks in  part to the principle of deferred measurement, we get
away with this.} It is a notable feature of the example above that the entanglement
needed only one extra qubit, and its measurement is something that happens deep inside our
quantum computer, showing that the theory that any collapse of the wave function requires a conscious
observer is misconceived. A quantum computer with functioning error correction
will be the strongest test to date that 
the view of  quantum theory, informally presented in these notes, 
with its strange mixture of unitary evolution punctuated by projections collapsing the quantum state
onto an eigenbasis of a Hermitian operator, 
is an accurate model for how the universe works.\footnote{\textbf{Reference:} 
it could be that quantum theory is merely a remarkably accurate model or --- and I think
this is the mathematically appealing alternative --- that quantum states
really are the basic `ontic' elements of reality. See \cite{Maudlin} for discussion of this
question.}

\subsection*{Acknowledgements}
I thank Michael D., Alastair Kay, and Tony Skyner for helpful comments on an earlier version of these notes.

\def\cprime{$'$} \def\Dbar{\leavevmode\lower.6ex\hbox to 0pt{\hskip-.23ex
  \accent"16\hss}D} \def\cprime{$'$}
\providecommand{\bysame}{\leavevmode\hbox to3em{\hrulefill}\thinspace}
\providecommand{\MR}{\relax\ifhmode\unskip\space\fi MR }
% \MRhref is called by the amsart/book/proc definition of \MR.
\providecommand{\MRhref}[2]{%
  \href{http://www.ams.org/mathscinet-getitem?mr=#1}{#2}
}
\providecommand{\href}[2]{#2}

\end{document}